%% file: main.tex
\pdfoutput=1 
\documentclass[sigconf]{acmart}
\input{macros}

\AtBeginDocument{%
  \providecommand\BibTeX{{%
    \normalfont B\kern-0.5em{\scshape i\kern-0.25em b}\kern-0.8em\TeX}}}

\setcopyright{acmcopyright}
\copyrightyear{2022}
\acmYear{2023}
\acmDOI{XXXXXXX.XXXXXXX}

\acmConference[Sigmod '23]{Make sure to enter the correct
  conference title from your rights confirmation emai}{June 18--23,
  2023}{Seattle, WA}
\acmPrice{15.00}
\acmISBN{978-1-4503-XXXX-X/18/06}

\begin{document}

%%
%% The "title" command has an optional parameter,
%% allowing the author to define a "short title" to be used in page headers.
\title{SafeBound: A Practical System for Generating Cardinality Bounds}

%%
%% The "author" command and its associated commands are used to define
%% the authors and their affiliations.
%% Of note is the shared affiliation of the first two authors, and the
%% "authornote" and "authornotemark" commands
%% The "author" command and its associated commands are used to define the authors and their affiliations.
\author{Kyle Deeds}
\affiliation{%
  \institution{University of Washington}
  \city{Seattle}
  \state{Washington}
  \country{United States}
}
\email{kdeeds@cs.washington.edu}

\author{Dan Suciu}
\affiliation{%
  \institution{University of Washington}
  \city{Seattle}
  \state{Washington}
  \country{United States}
}
\email{suciu@cs.washington.edu}

\author{Magda Balazinska}
\affiliation{%
  \institution{University of Washington}
  \city{Seattle}
  \state{Washington}
  \country{United States}
}
\email{magda@cs.washington.edu}

%%
%% By default, the full list of authors will be used in the page
%% headers. Often, this list is too long, and will overlap
%% other information printed in the page headers. This command allows
%% the author to define a more concise list
%% of authors' names for this purpose.
\renewcommand{\shortauthors}{Deeds, Suciu, Balazinska}
%%
%% The code below is generated by the tool at http://dl.acm.org/ccs.cfm.
%% Please copy and paste the code instead of the example below.
%%

%%
%% Keywords. The author(s) should pick words that accurately describe
%% the work being presented. Separate the keywords with commas.

%%
%% The abstract is a short summary of the work to be presented in the
%% article.

%\input{response}

\input{0-abstract}

\maketitle

\input{1-introduction}

\input{2-background}
\input{3-safebound}

\input{4-optimizations}

\balance

\input{5-evaluation}

\input{6-conclusions}

%\clearpage
\newpage
\bibliographystyle{ACM-Reference-Format}
\bibliography{references}

\newpage
%\appendix
%\input{appendix}

\end{document}

%% file: macros.tex
\usepackage{amsmath}
\usepackage{algorithm}
\usepackage{algorithmic}
\usepackage{subcaption}
\usepackage{bm} % bold math fonts, e.g. $\bm \Sigma$ will give a bold-face $\Sigma$
\usepackage{amsfonts}
\usepackage{amsthm}
\usepackage{graphicx}
\usepackage{tikz}
\usepackage[colorinlistoftodos]{todonotes}

\usepackage{enumitem}
\usepackage{wrapfig}
\usepackage{nicematrix}
\usepackage[bottom]{footmisc}

\theoremstyle{plain}                   % default
\newtheorem{thm}{Theorem}[section]
\newtheorem{lmm}[thm]{Lemma}
\newtheorem{prop}[thm]{Proposition}

\newtheorem{cor}[thm]{Corollary}

\theoremstyle{definition}              % Examples and all

\newtheorem{opm}{Question}
\newtheorem{conj}[thm]{Conjecture}

\newtheorem{defn}[thm]{Definition}

\newtheorem{rmk}[thm]{Remark}

\newcommand{\bi}{\begin{itemize}}
\newcommand{\ei}{\end{itemize}}

\newcommand{\bp}{\begin{proof}}
\newcommand{\ep}{\end{proof}}
\newcommand{\bcor}{\begin{cor}}
\newcommand{\ecor}{\end{cor}}
\newcommand{\bthm}{\begin{thm}}
\newcommand{\ethm}{\end{thm}}
\newcommand{\blmm}{\begin{lmm}}
\newcommand{\elmm}{\end{lmm}}
\newcommand{\bdefn}{\begin{defn}}
\newcommand{\edefn}{\end{defn}}
\newcommand{\bprop}{\begin{prop}}
\newcommand{\eprop}{\end{prop}}
\newcommand{\bconj}{\begin{conj}}
\newcommand{\econj}{\end{conj}}
\newcommand{\bopm}{\begin{opm}}
\newcommand{\eopm}{\end{opm}}
\newcommand{\brmk}{\begin{rmk}}
\newcommand{\ermk}{\end{rmk}}

\newcommand{\defeq}{\stackrel{\text{def}}{=}}

\newcommand{\eat}[1]{}

\newcommand{\magda}[1]{\textcolor{blue}{[[Magda: #1]]}}

\newcommand{\dan}[1]{\textcolor{red}{[[Dan: #1]]}}

%% file: 0-abstract.tex
\begin{abstract}
Recent work has reemphasized the importance of {\em cardinality estimates} for query optimization.  While new techniques have continuously improved in accuracy over time, they still generally allow for under-estimates which often lead optimizers to make overly optimistic decisions. This can be very costly for expensive queries.  An alternative approach to estimation is {\em cardinality bounding}, also called pessimistic cardinality estimation, where the cardinality estimator provides guaranteed upper bounds of the true cardinality. By never underestimating, this approach allows the optimizer to avoid potentially inefficient plans. However, existing pessimistic cardinality estimators are not yet practical: they use very limited statistics on the data, and cannot handle predicates.  In this paper, we introduce SafeBound, the first practical system for generating cardinality bounds.  SafeBound builds on a recent theoretical work that uses degree sequences on join attributes to compute cardinality bounds, extends this framework with predicates, introduces a practical compression method for the degree sequences, and implements an efficient inference algorithm.  Across four workloads, SafeBound achieves up to 80\% lower end-to-end runtimes than PostgreSQL, and is on par or better than state of the art ML-based estimators and pessimistic cardinality estimators, by improving the runtime of the expensive queries.  It also saves up to 500x in query planning time, and uses up to 6.8x less space compared to state of the art cardinality estimation methods.
\end{abstract}

%% file: 1-introduction.tex
\section{Introduction}
\label{sec:intro}

When considering how to execute a database query, the
query optimizer relies on cardinality
estimates to determine costs of potential plans and choose an efficient one. Recent benchmarks
of query optimizers have shown that (1) traditional cardinality
estimators routinely underestimate the true cardinality by many orders
of magnitude~\cite{DBLP:journals/vldb/LeisRGMBKN18} (2)
underestimating cardinalities can result in highly inefficient query
plans~\cite{DBLP:journals/vldb/LeisRGMBKN18,DBLP:journals/pvldb/HanWWZYTZCQPQZL21,DBLP:conf/sigmod/KimJSHCC22}
and (3) accurate estimates of large (sub) queries are crucial to
generating good query
plans~\cite{DBLP:journals/pvldb/HanWWZYTZCQPQZL21}.

These observations have motivated work on \textit{cardinality
  bounding}, also called pessimistic cardinality estimation, where the
cardinality estimator provides guaranteed upper bounds of the true
cardinality rather than unbiased
estimates~\cite{DBLP:conf/sigmod/CaiBS19,DBLP:conf/cidr/HertzschuchHHL21}. When
given a cardinality bound rather than a cardinality estimate and
assuming a reasonably accurate cost model, the query optimizer will
avoid the most inefficient query plans, e.g. choosing a
nested loop join between two large tables.  Formulas that bound the
output cardinality of a query have been introduced in the theoretical
community~\cite{DBLP:conf/soda/GroheM06,DBLP:journals/jacm/GottlobLVV12,DBLP:conf/pods/KhamisNS16,DBLP:conf/pods/Khamis0S17,DBLP:journals/corr/abs-2201-04166},
and implementations based on simplified versions of these formulas are
described in~\cite{DBLP:conf/sigmod/CaiBS19} and
in~\cite{DBLP:conf/cidr/HertzschuchHHL21}.  One comprehensive study of
cardinality estimation methods reports good end-to-end performance by cardinality bounds compared to other
methods~\cite{DBLP:journals/pvldb/HanWWZYTZCQPQZL21}. Recent work has even found novel connections between these cardinality bounds and more traditional estimators, allowing cardinality bounding optimizations to be applied to traditional estimates \cite{DBLP:journals/pvldb/ChenHWSS22}. However, current techniques for cardinality bounds produce loose estimates and  do not support predicates in queries. In this work, we overcome these limitations and describe, SafeBound, a practical system that uses a rich collection of statistics to generate bounds and supports a wide range of predicates.

% Further, queries that have
% a large true cardinality tend to be close to their cardinality bound,
% so this approach naturally results in high accuracy on the most
% expensive queries ~\cite{DBLP:conf/sigmod/CaiBS19}. An additional,
% subtle benefit of cardinality bounds is their ease of
% composition. While choosing between a variety of unbounded cardinality
% estimates is a complex problem with no standard solution, the best
% cardinality bound from an ensemble of bounds is simply the lowest one.
% \magda{Other comments: Please add citations for each type of
%   category. Avoid "badly" and "very" ("very" can simply be removed).}

To place SafeBound in context, we briefly review the landscape of
cardinality estimation techniques.  Prior work falls in three broad
categories: traditional methods, ML methods, and cardinality
bounding methods. Traditional methods are supported by all major
database systems and rely on histograms, distinct counts, and
most-common-value lists, combined with strong assumptions (uniformity,
independence, inclusion of values). These methods are fast
and space efficient but suffer from the drawbacks listed above. Also
under this umbrella are sampling-based methods, which cope with
non-uniformity and correlations, but use indexes during estimation and
do not scale to queries with many joins.  To compensate for this, a
great deal of recent work has gone into creating increasingly complex
ML models that automatically detect these
correlations~\cite{DBLP:journals/pvldb/YangKLLDCS20,DBLP:conf/cidr/KipfKRLBK19,DBLP:journals/pvldb/WangQWWZ21,DBLP:journals/pvldb/LiuD0Z21,DBLP:journals/pvldb/ZhuWHZPQZC21,DBLP:conf/sigmod/KimJSHCC22, DBLP:conf/sigmod/WoltmannHTHL19,DBLP:conf/sigmod/WoltmannHHL20,DBLP:conf/sigmod/WuC21,DBLP:journals/pvldb/WangCLL21} . While
this is certainly an exciting approach and has yielded high accuracy
estimates on existing benchmarks, many practical hurdles remain: they
require a large memory footprint, have a slow training time, and are hard
to
interpret~\cite{DBLP:journals/pvldb/WangQWWZ21,DBLP:conf/sigmod/KimJSHCC22,DBLP:journals/pvldb/HanWWZYTZCQPQZL21}.
Further, these approaches generally treat under and overestimation
identically, which can result in a disconnect between their accuracy
and their effect on workload
runtime~\cite{DBLP:journals/pvldb/NegiMKMTKA21,DBLP:conf/sigmod/KimJSHCC22}.

Previous efforts in cardinality bounding have resulted in good workload runtimes. However, they have been either unable to handle predicates and slow to perform inference or inaccurate and unable to produce guaranteed bounds. PessEst ~\cite{DBLP:conf/sigmod/CaiBS19} hash-partitions each relation and calculates its bound using the cardinality and max degree for each partition. In order to handle predicates, they resort to scanning base tables at inference time, an unacceptable overhead for query optimizers. Further, the inference time grows exponentially in the total number of partitions. Simplicity~\cite{DBLP:conf/cidr/HertzschuchHHL21} uses the same bound as~\cite{DBLP:conf/sigmod/CaiBS19}, but without any hash refinement. Instead, they focus on producing single-table estimates through sampling.  However, this results in loose bounds and no longer produces a guaranteed upper bound on the query's cardinality (see Fig.~\ref{subfig:relative-error}).

Our work draws inspiration from a recent theoretical
result~\cite{DBLP:journals/corr/abs-2201-04166}, which describes an
upper bound formula using the {\em degree sequences} of all join
attributes; this is called the \textit{Degree Sequence Bound}, DSB.
The degree sequence of an attribute is the sorted list
$f_1 \geq f_2 \geq \cdots \geq f_d$ of the frequencies of distinct
values in that column: we illustrate it briefly in
Fig.~\ref{fig:degree_sequence} and define it formally in
Sec.~\ref{sec:problem:definition}.  The degree sequence captures a
rich set of statistics on the relation: the cardinality is
$\sum_i f_i$, the maximum frequency is $f_1$, the number of distinct
values is $d$, etc.  The memory footprint of the full degree sequence
is too large, but it can be compressed using a piecewise
constant function $\hat f$ such that $f \leq \hat f$ (see
Fig.~\ref{fig:degree_sequence}), thus the technique offers a tradeoff
between memory and accuracy.  The authors
in~\cite{DBLP:journals/corr/abs-2201-04166} describe a mathematical
upper bound on output cardinality given in terms of the compressed
degree sequences.  However, the theoretial formalism is not practical:
it ignores predicates; the proposed compression definition
artificially increases the cardinality of the relation; and there is
no concrete compression algorithm.

In this paper, we describe SafeBound, the first practical system for
generating cardinality bounds from degree sequences.  Like any
cardinality estimator, SafeBound has an offline and an online phase
(Sec.~\ref{sec:overview}).  During the offline phase it computes
several degree sequences, then compresses them. 
Unlike~\cite{DBLP:journals/corr/abs-2201-04166}, we also compute
degree sequences {\em conditioned on predicates}: SafeBound supports
equality predicates, range predicates, LIKE, conjunctions, and
disjunctions (Sec.~\ref{subsec:choosing:statistics}).  Next, we
describe a stronger compression method by upper bounding the {\em
  cumulative} degree sequence rather than the degree sequence, and
prove that it preserves the cardinalities of the base tables, yet
still leads to a guaranteed upper bound on the query's output
cardinality (Sec.~\ref{subsec:cdsb}).  Any compression necessarily 
looses some precision, and no compression is optimal for all queries. 
Thus, we describe a heuristic-based compression algorithm in Sec.~\ref{subsec:modeling:cdfs}.
At the end of the offline phase of
SafeBound we have a set of compressed degree sequences.  Next,
during the online phase, SafeBound takes a query consisting of joins
and predicates, and computes a guaranteed upper bound on its output
cardinality, using the compressed degree sequences.  For this purpose
we implemented an algorithm (inspired
by~\cite{DBLP:journals/corr/abs-2201-04166}) that computes the upper
bound in almost-linear time in the total size of all compressed degree
sequences.  In addition to these fundamental techniques, we describe
several optimizations in Sec.~\ref{sec:optimizations}.  Finally, we
evaluate SafeBound empirically in Sec.~\ref{sec:evaluation}. Across
four workloads, SafeBound achieves up to 80\% lower end-to-end
runtimes than PostgreSQL, and is on par or better than state of the
art ML-based estimators and pessimistic cardinality estimators.  Its
performance gains come especially from the expensive queries, because
its guarantees on the cardinality prevent the optimizer from making
overly optimistic decisions. In the tradeoff between accuracy and build time, SafeBound aims for accuracy. This results in slower build times than Postgres although it is 2-20x faster to build than state of the art ML methods. SafeBound also saves up to 500x in query planning time, and uses up to 6.8x less space compared to state of the
art cardinality estimation methods.

\noindent In summary, this paper makes the following contributions:
\begin{itemize}
\item \textit{SafeBound Design:} We describe the architecture of
  SafeBound, the first practical system for cardinality bounds, in
  Sec.~\ref{sec:overview}.
\item \textit{Predicates:} We introduce a scheme for conditioning degree sequences on predicates (Sec.~\ref{subsec:choosing:statistics}).
    \item \textit{Valid Compression:} We describe an improved
      compression of degree sequences that preserves the cardinality
      of the base tables and still leads to a guaranteed upper
      bound (Sec.~\ref{subsec:cdsb}).
    \item \textit{Compression Algorithm:} We present an algorithm for
      compressing a degree sequence with minimal loss (Sec.~\ref{subsec:modeling:cdfs}).
      %%%%%% I (Dan) commented this out.  I don't think we need to
      %%%%%% list it as a contribution.
%     \item \textit{Upper Bound Algorithm:} We implement an algorithm
%       for fast computation of the degree sequence bound. Sec. ~\ref{subsec:fdsb}.
    \item \textit{Optimizations:} We describe several optimizations in
      Sec.~\ref{sec:optimizations}.
    \item \textit{Experimental Evaluation:} We perform a thorough
      experimental evaluation on the JOB-Light, JOB-LightRanges,
      JOB-M, and STATS-CEB benchmarks demonstrating SafeBound's fast inference, low memory overhead, and nearly optimal workload runtimes across all benchmarks. Sec.~\ref{sec:evaluation}
\end{itemize}

%% file: 2-background.tex
\begin{figure}[t]
    \centering
    \includegraphics[width=.45\textwidth, keepaspectratio]{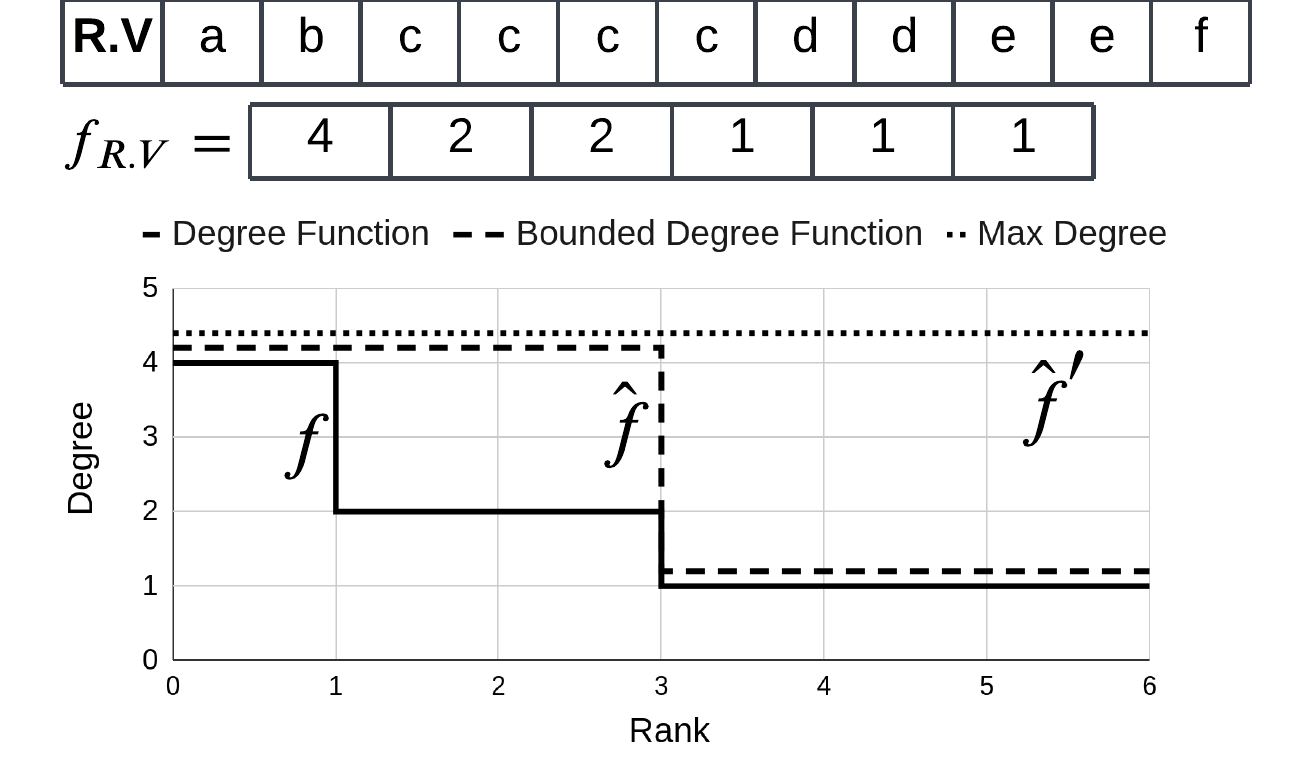}
    \caption{An example column, $R.V$, and its degree sequence
      $f$. The degree sequence can be compressed by upper bounding it
      with a step function.  We show two step functions, $\hat f$ with
      2 segments and $\hat f'$ with 1 segment:
      $f \leq \hat f \leq \hat f'$.  Notice that they will
      overestimate the total cardinality of the relation, from the
      true cardinality 11, to 15, and to 24, respectively.}
    \label{fig:degree_sequence}
\end{figure}

\section{Problem Definition and Background}
\label{sec:problem:definition}

\begin{figure}[t]
    \centering
    \includegraphics[width=.43\textwidth,keepaspectratio=True]{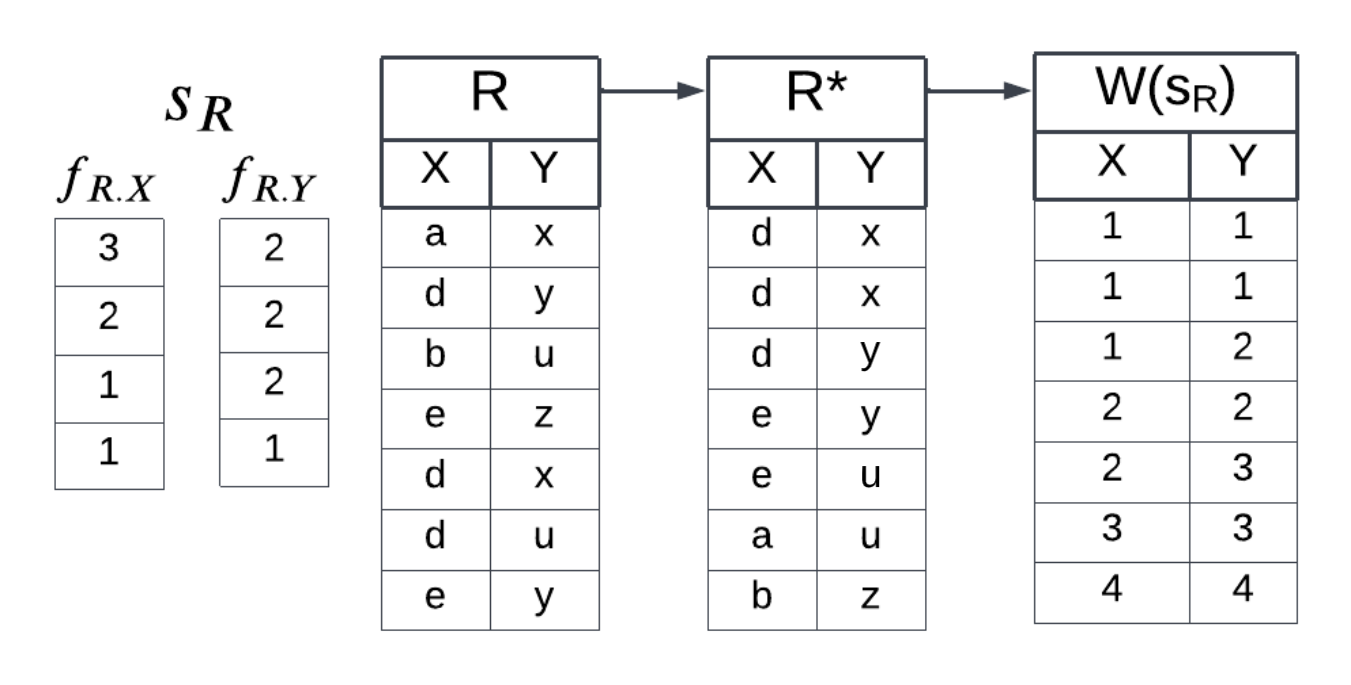}   \caption{This figure shows the conversion of an input table $R(X,Y)$ to the worst-case instance $W(s_R)$ which has the same degree sequences but produces larger join outputs. This requires first sorting the columns individually by frequency then re-labeling join values in order of frequency.}
    \label{fig:worst:case:instance}
\end{figure}

\subsection{The Cardinality Bounding Problem}
\label{subsec:problem}
Let
$Q = R_1(\bm V_1) \wedge \ldots \wedge R_k(\bm V_k) \wedge P_1 \wedge
\ldots \wedge P_l$ be a query, expressed using datalog notation, over
relations $R_i\in\bm R$ with variables $\bm V_i\subseteq\bm V$ and
sets of predicates $\bm P_1\ldots \bm P_l$. To match the setting of \cite{DBLP:journals/corr/abs-2201-04166}, we generally assume that queries are acyclic and that any join between two relations
is on a single attribute. These queries, called Berge-acyclic queries,
include many natural queries\footnote{On closer examination, the queries supported
  by~\cite{DBLP:conf/sigmod/CaiBS19,DBLP:conf/cidr/HertzschuchHHL21}
  are also Berge-acyclic.}, like chain, star, snowflake
queries~\cite{fagin1983degrees}. However, we go on to provide extensions to handle 
general queries as well in Sec. \ref{subsec:cyclic-multi-key}. We further assume that $Q$ is a full
conjunctive query, i.e. $\texttt{select *}$ in SQL,
and uses bag semantics: both the input relations and the
output are bags.  Our goal is to compute a bound on its output.

Let $D$ be a database instance with statistics $s$. The goal of
cardinality bounding is to generate a bound $B(Q, s)$ such that the
following is true,
\begin{align}\label{def:cardinality:bounding}
\max_{D\models s} |Q(D)| \leq B(Q, s)
\end{align}
The bound, $B(Q,s)$, guarantees that the query $Q$ will have true
cardinality less than $B(Q, s)$ when executed on any database
consistent with the given statistics, $s$, which we denote $D \models
s$. We say that a bound $B(Q,s)$ is \textit{tight} if there exists a
database $D$ consistent with $s$ for which the inequality
\eqref{def:cardinality:bounding} is an equality. In this paper, $s$
will specifically refer to the degree sequences.

\subsection{The Degree Sequence}
The core statistic in our system, SafeBound, is the degree sequence
(DS), so we describe it in detail here. Consider a variable $V$
(a.k.a. column) of a relation $R$ with domain $\mathbb{D}_V$ and
values $v\in \mathbb{D}_V$. The frequency of a particular value $v$ in
relation $R$ is $|\sigma_{V=v}(R)|$. Ranking the values in descending
order by frequency, we get a list of values
$v^{(1)},\ldots,v^{(|\mathbb{D}_V|)}$ whose frequencies,
$f_{R.V,i} \defeq |\sigma_{V=v^{(i)}}(R)|$, form the \textit{degree
  sequence} of $R.V$: $f_{R.V,1} \geq f_{R.V,2} \geq \cdots$ The index
$i$ is called the {\em rank}.  The {\em cumulative degree sequence},
CDS, is the running sum of the DS, i.e.
$F_{R.V,i} \defeq \sum_{j=1}^if_{R.V,j}$. In the other direction, the
DS can be defined as the \textit{discrete derivative} of the CDS,
$f_{R.V,i} = \Delta_i F_{R.V,i} = F_{R.V,i}-F_{R.V,i-1}$.

In this paper, we often manipulate these sequences as piecewise functions and may use the following notation, $f_{R.V}(i)= f_{R.V,i}$ and
$F_{R.V}(i) = F_{R.V,i}$, for the DS and CDS, respectively. The set of DSs for a relation $R$ will be denoted by $s_R$, and the set of DSs for all relations by $s$; similarly, $S_R$ and $S$ stand for the CDSs of one relation, or all
relations.  Let $s, \hat s$ be two sets of statistics; we write
$s \leq \hat s$ to mean that $f_{R.V}(i) \leq \hat f_{R.V}(i)$ for all
relations $R$, attributes $V$, and ranks $i$.  Given a set of
statistics $\hat s$ and a database $D$, we say that $D$ is
\textit{consistent} with $\hat s$ if $s(D) \leq \hat s$, where $s(D)$
are the degree sequences for the instance $D$.  An example of these
concepts can be seen in Figure \ref{fig:degree_sequence}. At the top
are the actual values of a column $R.V$ of the database, and below is
the degree sequence: the value $4$ corresponds to $c,c,c,c$, the
values $2,2$ correspond to $d,d$ and $e,e$, etc.  The graph shows the
degree sequence $f$ as a solid line.  In general, the degree sequence
can be large.  For the purpose of computing an upper bound of the size
of the query, it suffices to compress the DS by upper bounding it
using a piecewise constant function: two such functions are shown in
the figure, $f \leq \hat f \leq \hat f'$.

\subsection{The Degree Sequence Bound} \label{subsec:dsb}
Previous work~\cite{DBLP:journals/corr/abs-2201-04166} described a
method to compute a tight upper bound on the query's output from the
degree sequences, called the {\em degree sequence bound}, DSB, which
we briefly describe here.

\noindent\textbf{The Worst-Case Instance:} The DSB is defined in terms
of a \textit{worst-case relation}, $W(s_R)$, associated with the
statistics $s_R$, and the \textit{worst-case instance} of the
database, $W(s)$, consisting of all relations $W(s_R)$.  These
relations are defined such that the size of any query's answer on
$W(s)$ is an upper bound of its size on any database consistent with
$s$; in this sense, $W(s)$ is the ``worst'' instance.  Because degree
sequences do not capture the correlation between columns within
relations or between relations, the worst-case assumption is that the
frequency of values is perfectly correlated across columns both within
and between relations: i.e., high frequency values occur in the same
tuples within relations and join with high frequency values across
relations.

We illustrate how to produce a worst-case relation $W(s_R)$ from a
binary relation $R(X,Y)$ in
Figure~\ref{fig:worst:case:instance}. While we are given only the
statistics $s_R$, i.e.  the two degree sequences $f_{R.X}, f_{R.Y}$,
we also show an instance $R$, to help build some intuition into
$W(s_R)$.  First, we sort each column independently by frequency to
produce $R^*$. This ensures that high-frequency values occur in the
same tuples. Next, we relabel our join values in order of frequency to
produce $W(s_R)$, i.e. $x^{(1)}=1,x^{(2)}=2,\ldots$. This ensures
that, when we join, say, $R$ and $T$, the high-frequency values in
$W(s_R)$ will join with the high-frequency values in $W(s_T)$.
% At this point, we can represent each column as a function from
% values in the domain, e.g. in $[1,|\mathbb{D}_X|]$, to row indices,
% in $[1,\ldots,|R|]$. Specifically, this function is the CDF of the
% column.

\noindent\textbf{The Degree Sequence Bound:}
The DSB is the size of the query $Q$ run on the worst-case instance
$W(s)$, in other words, $|Q(W(s))|$. The following is shown
in~\cite{DBLP:journals/corr/abs-2201-04166}:

\begin{thm} Suppose $Q$ is a Berge-acyclic query, and let $s$ be a set of degree sequences, one for each attribute of each relation. Then, the following is true,
\begin{align}
    \forall D \models s\quad |Q(D)| &\leq  |Q(W(s))| \label{eq:dsb:icdt}
\end{align}
Further, $W(s) \models s$, which proves that the bound
in~\eqref{eq:dsb:icdt} is tight.
\label{th:dsb}
\end{thm}
The theorem forms the theoretical basis for compression.  Suppose that
$\hat s$ is a compressed (lossy) representation of $s$, such that
$s \leq \hat s$.  For example, $\hat s$ may replace the degree
sequences with piecewise constant functions with a small number of
segments, as in Fig~\ref{fig:degree_sequence}.  Then $|Q(W(\hat s))|$
is still an upper bound for $Q(D)$, because $D \models s$ implies
$D \models \hat s$, and we apply the theorem to $\hat s$.

\subsection{Discussion} \label{subsec:dsb:discussion}

In summary, the theory in~\cite{DBLP:journals/corr/abs-2201-04166}
introduced the DSB, and suggested compressing the degree sequences;
they also described an efficient algorithm for the DSB, which we
implemented (reviewed in Sec.~\ref{subsec:fdsb}).  Using degree
sequences for cardinality bounding is attractive, because they capture
several popular statistics used in cardinality estimation: the number
of distinct values in a column $R.V$ is $\|f_{R.V}\|_0$; the
cardinality of the relation $R$ is $\|f_{R.V}\|_1$; and the maximum
degree (or maximum frequency) is $\|f_{R.V}\|_\infty$.  However, the
framework in~\cite{DBLP:journals/corr/abs-2201-04166} is impractical,
for several reasons:
\begin{enumerate}
\item The DSB in Theorem~\ref{th:dsb} does not take predicates into
  account; if the query contains some predicate, e.g. $R.B=2$, then
  the DSB simply ignores it, which means that the upper bound in
  Equation~\eqref{eq:dsb:icdt} will be a huge overestimate.
\item If the compressed sequences $\hat s$ dominate $s$, i.e.
  $s \leq \hat s$, then they artificially increase the estimated
  cardinality of the relations, which increases the DSB.
\item While it describes a bound based on compressed degree sequences, no concrete algorithm for compressing degree sequences is suggested in~\cite{DBLP:journals/corr/abs-2201-04166}.
\end{enumerate}

In this paper, we make several contributions to address these
limitations, and describe SafeBound, a practical cardinality bounding
system.

%% file: 3-safebound.tex
\section{SafeBound} \label{sec:safebound}

In this section, we present SafeBound, the first practical system for
cardinality bounding, which is based on several extensions of the results in~\cite{DBLP:journals/corr/abs-2201-04166}.  We start with an
overview of SafeBound.

\subsection{Overview} \label{sec:overview}

SafeBound has an offline and an online phase.  During the {\em offline
  phase} the system computes the degree sequences of every join-able
attribute of every relation (keys and foreign keys). In addition it
also considers a variety of predicate types on each relation, and
computes refined degree sequences conditioned on those predicates:
this is described in Sec.~\ref{subsec:choosing:statistics}.  The last
step of the offline phase consists of compressing the degree
sequences; this is described in Sec.~\ref{subsec:modeling:cdfs}. Rather than directly compressing the degree sequence as allowed by~\cite{DBLP:journals/corr/abs-2201-04166}, it compresses the cumulative degree sequence which does not inflate the
cardinality of the relation and requires a new correctness proof.
This is described in Sec.~\ref{subsec:cdsb}.  During the {\em online
  phase}, SafeBound receives a query, as defined in
Sec.~\ref{sec:problem:definition}, and computes an upper bound on the query's output using the pre-computed compressed degree sequences.  It does not apply
the formula in Theorem~\ref{th:dsb} naively, but, instead, it
implements a fast algorithm that runs in time proportional to the
total size of the compressed sequences; this is described in
Sec.~\ref{subsec:fdsb}.

\begin{example}\label{ex:running:example}
  We will refer to the following running example:
\begin{align*}
  Q \defeq & R(X,A,B)\wedge S(X,Y,C) \wedge T(Y) \\
    \wedge & (A<5) \wedge (B=2) \wedge (C\,\,\texttt{LIKE  '\%Abdul\%'})
\end{align*}
During the offline phase (before the query arrives) SafeBound computes
the degree sequences for $R.X, R.A, \ldots, T.Y$, as well as degree
sequences conditioned on predicates, e.g. degree sequences of
$R.X$ conditioned on range predicates applied to $R.A$.  All these degree sequences are compressed and stored. During the online phase,
SafeBound takes the query $Q$ above, and uses all available degree
sequences to compute an upper bound to the query's output.
\end{example}

We will use the following terminology.  A column to which a predicate
is applied is called a \textit{filter column}; a column used in a join
is called a \textit{join column}.  Note that a column can be both a
filter column and a join column.

\begin{figure*}
    \centering
    \includegraphics[width=.88\textwidth, height=110pt]{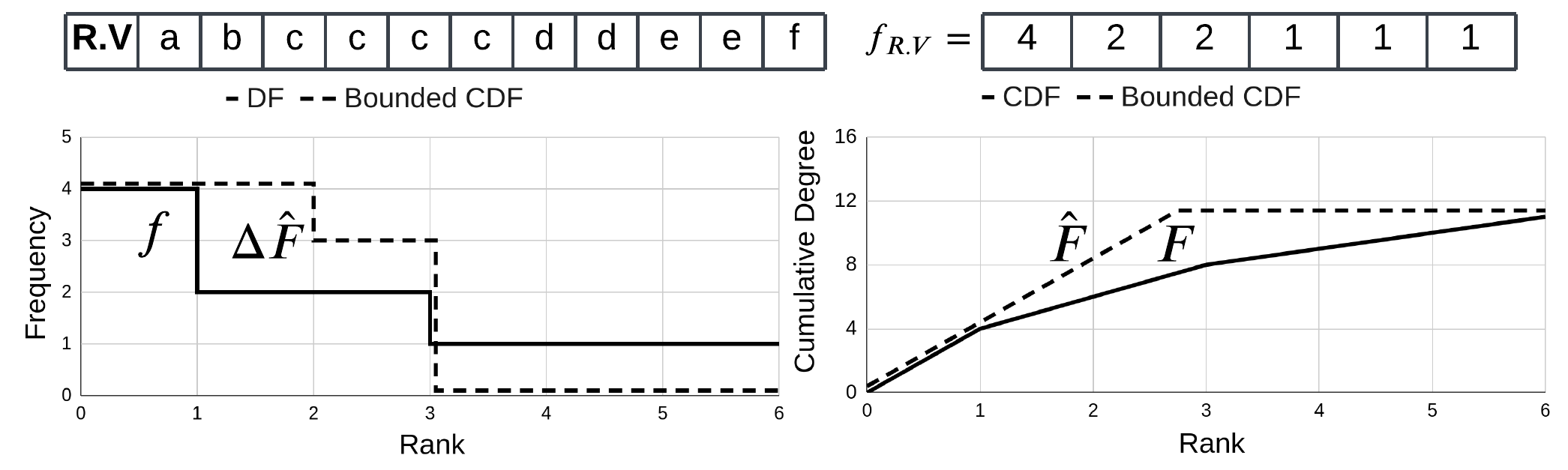}
    \caption{We show here how to compress the CDS of the column
      Fig.~\ref{fig:degree_sequence}.  The degree sequence $f$ on the
      left has the cumulative degree sequence $F$ on the right.  We
      upper bound the latter by $\hat F$, a piecewise linear
      function with two segments.  Notice that the cardinality of the
      relation is preserved: $|R|=F(6)=\hat F(6)=11$.  The degree sequence $\hat f \defeq \Delta \hat F$ associated to $\hat F$ no longer dominates the original $f$. Yet, Theorem~\ref{thm:cdf:upper:bound} proves that we can still use $\hat f$ to compute an upper bound on the
      cardinality of a query.}
    \label{fig:cdf:diagram}
\end{figure*}

\subsection{Conditioning on  Predicates}\label{subsec:choosing:statistics}

Predicates in a query can significantly reduce the cardinality of the
output.  SafeBound accounds for predicates by computing additional
degree sequences for each join column by {\em conditioning} on
predicates.  SafeBound supports five types of predicates: equality,
range, conjunctions, LIKE, and disjunctions. In this section we assume
that all degree sequences are exact; once we replace them with
compressed sequences in the next section, we will discuss some minor
adjustments to the formulas introduced here.

\noindent\textbf{Equality Predicates:} The main idea is the
following.  Let $V$ be a join column of a relation $R$, and consider a
query that includes an equality predicate, $R.A = a_k$, for some
constant $a_k \in \mathbb{D}_A$.  If we ignore the predicate and
compute the upper bound of the query without this predicate, then we
massively overestimate the query's cardinality.  There are two extreme
ways to account for the equality predicate.  At one extreme, we could
compute a separate degree sequence, $f_{R.V|(A=a_\ell)}$, for each value $a_1, a_2, \ldots \in \mathbb{D}_A$ from the subset
$\sigma_{A=a_\ell}(R)$.  At query estimation time, we use the degree
sequence $f_{R.V|(A=a_k)}$ instead of the unconditioned sequence
$f_{R.V}$.  But this approach is prohibitive, because it requires
storing a large number of degree sequences.  At the other extreme, we
could store only the max of all these individual degree sequences,
\begin{align}
f_{R.V|A} \defeq & \max_\ell f_{R.V|(A=a_\ell)} \label{eq:max:eq}
\end{align}
and use this to estimate the query's upper bound.  This uses very
little space, only one DS which we use for every equality predicate
on $A$, no matter the constant $a_k$.  But this may lead to
significant over-approximations of the upper bound.  SafeBound adopts
a compromise.  For each attribute $A$, it computes a Most Common Value
(MCV) list, computes separate degree sequences for $V$ conditioned on
each of these values, and computes one default DS, given by
formula~\eqref{eq:max:eq}, where $a_\ell$ ranges over non-MCV values.
At estimation time, if $a_k$ is in the MCV list then we use its own
degree sequence, otherwise we use the default one.  In our system, we
chose between $1000$ and $5000$ values to include in every MCV list.  We
will slightly revisit Eq.~\eqref{eq:max:eq} in the next section.

\noindent\textbf{Range Predicates:} To handle range predicates,
SafeBound uses a data structure that builds on the idea of traditional
histograms. The naive way to do this is to compute an equi-depth
histogram over $R.A$ where each bucket stores the degree sequence of
$R.V$ restricted to that bucket. However, a range predicate may
overlap multiple buckets which would require us to return the
summation of the degree sequences in those buckets. This summation
artificially inflates the highest frequency values, increasing the DSB significantly.
Instead, we build a hierarchy of equi-depth histograms with
$2^k,2^{k-1},\ldots,2^1$ buckets. At query time, we identify the
smallest bucket which fully encapsulates the range query and return
the degree sequence stored there. In our system, we typically used
$k=7$.

\noindent\textbf{Conjunctions:} Suppose that the query contains
a conjunction of predicates on the same relation $R$:
$P_1(R.A) \wedge P_2(R.B) \wedge \cdots$ For each predicate, we have
computed a degree sequence conditioned on that predicate,
$f_{R.V|A}, f_{R.V|B}, \ldots$ Then, we take their minimum:
\begin{align}
f_{R.V|(A \wedge B \wedge \cdots)}(i) \defeq \min (f_{R.V|A}(i), f_{R.V|B}(i), \ldots) \label{eq:min:conjunction}
\end{align}
Referring to our running Example~\ref{ex:running:example}, there are
two predicates on $R$.  The degree sequence $f_{R.X}$ conditioned on
the conjunction is
$f(i) \defeq \min(f_{R.X|(A<5)}(i), f_{R.X|(B=2)}(i))$.

\noindent\textbf{LIKE Predicates:} 
SafeBound converts a predicate $R.A\ \texttt{LIKE'\%xyzu\%'}$ into
a conjunction of predicates on 3-grams.  For every attribute
$R.A$ of type \texttt{text}, SafeBound first computes an MCV list of 3-grams that occur in $A$ and calculates the degree sequences of $R.V$ conditioned on all 3-grams that occur in the MCV. Separately, it calculates the degree sequence of $R.V$ conditioned on $A$ not containing any 3-gram in the MCV.  At query time, we split the text in the \texttt{LIKE} predicate into 3-grams, then compute the $\min$ of all degree sequences conditioned on each 3-gram that occurs in the MCV list. If none of the predicate's 3-grams appear in this list, we use the degree sequence conditioned on not containing common 3-grams. Referring to our running Example~\ref{ex:running:example}, the string \texttt{'\%Abdul\%'} is split into the 3-grams \texttt{Abd}, \texttt{bdu}, \texttt{dul}; for each 3-gram we retrieve the degree sequence $S.X$ conditioned on that 3-gram, then compute their min (or take the default if none of them are in the MCV list); we apply the same process to compute the degree sequence of $S.Y$.

\noindent\textbf{Disjunctions:} Suppose a query has a disjunction of
predicates over a relation $R$,
$P_{1}(R.A_1)\vee\cdots\vee P_k(R.A_k)$.  For each predicate we have
the degree sequence of $R.V$ conditioned on it; then, we take their
sum.  For example, the \texttt{IN} predicate in SQL is treated as a
disjunction: the degree sequence condition on of
$R.A \text{ IN ['a','b','c']}$ is
$f_{R.V|(A=a)}+f_{R.V|(A=b)}+f_{R.V|(A=c)}$.

\begin{example}
  The \texttt{Title} relation in the JOBLightRanges benchmark has 7
  filter columns (\texttt{episode\_nr}, \texttt{season\_nr},
  \texttt{production\_year}, \texttt{series\_years},
  \texttt{phonetic\_code}, \texttt{series\_years},
  \texttt{imdb\_index}) and two join columns (\texttt{id},
  \texttt{kind\_id}). This results in seven histograms, MCV lists,
  and, for the string attributes, 3-gram lists which store 2 degree
  sequences per bin, MCV, and tri-gram, respectively. In total, there
  are $18,522$ degree sequences for the relation \texttt{Title}, each describing a subset of the table. This
  motivates our compression technique in Sec.~\ref{subsec:cdsb}
  and~\ref{subsec:modeling:cdfs}, and the group compression in
  Sec.~\ref{subsec:grouping:CDS:sets}.
\end{example} 

\noindent\textbf{Discussion:} SafeBound computes DS's only for {\em join columns} (keys and foreign keys), each conditioned on every {\em filter column}.  In theory this could lead to $O(n^2)$ DS's for a table with $n$ attributes, but in practice, a typical table has $O(1)$ foreign keys, resulting in $O(n)$ DS’s per table.

% \dan{I believe we discuss this enough in Sec. 6}
% \revb{\noindent\textbf{Discussion:} The techniques presented here adapt traditional techniques for handling predicates, and they suffer from some of the same drawbacks. Due to their limited granularity, they have difficulty in distinguishing between highly selective predicates. In traditional systems, this results in highly inaccurate estimates on these predicates unless uniformity assumptions hold. In SafeBound, we cannot assume that the selectivity is less than the finest level of granularity provided by our statistics (e.g. the number of rows which fall in the smallest histogram bins), resulting in overestimates for the most selective predicates.}

\subsection{Valid Compression}\label{subsec:cdsb}

The degree sequence statistics, $s$, are as large as the database
instance, $D$, hence they are impractical for cardinality estimation.
SafeBound compresses each degree sequence using a piecewise constant
function with a small number of segments.  We denote by $\hat s$ the
collection of compressed degree sequences.  As we saw in
Theorem~\ref{th:dsb}, the theoretical results
in~\cite{DBLP:journals/corr/abs-2201-04166} required
$D \models \hat s$, meaning that that every degree sequence $f_{R.V}$
of the database $D$ is dominated by the corresponding DS in $\hat s$:
$f_{R.V}(i) \leq \hat f_{R.V}(i)$.  The problem is that, if we
increase the degree sequence $f_{R.V}$ to $\hat f_{R.V}$, then the
cardinality of the worst case relation $W(\hat s_R)$ will increase
artificially, from $\sum_i f_{R.V}(i)$ to $\sum_i \hat f_{R.V}(i)$. This leads to poor upper bounds.  In this section we introduce a
stronger compression method, which does not increase the cardinality
of the worst case instance.

% This compression technique requires some
% minor revisions of the conditional degree sequences in
% Sec.~\ref{subsec:choosing:statistics}, which we discuss at the end of
% this section.

Our new idea is that, instead of dominating the degree sequence,
$f_{R.V} \leq \hat f_{R.V}$, we will dominate the {\em cumulative}
degree sequence\footnote{Recall that
  $F_{R.V}(i) \defeq \sum_{j\leq i}f_{R.V}(i)$ and
  $\hat F_{R.V}(i) \defeq \sum_{j\leq i}\hat f_{R.V}(i)$.}
$F_{R.V}(i) \leq \hat F_{R.V}(i)$.  Obviously, if the DS is dominated,
then the CDS is dominated too, in other words
$f_{R.V} \leq \hat f_{R.V}$ implies $F_{R.V} \leq \hat F_{R.V}$, but
the converse does not hold in general, as is illustrated in
Fig.~\ref{fig:cdf:diagram}.  The advantage is that we can dominate the
CDS yet still preserve the cardinality, by ensuring
$|R| = F_{R.V}(|\mathbb{D}_V|) = \hat F_{R.V}(|\mathbb{D}_V|)$, see
Fig.~\ref{fig:cdf:diagram}, but the problem is that we can no longer
use Theorem~\ref{th:dsb} to conclude that the compressed sequence
leads to an upper bound.  We show that the upper bound still holds, by
proving the new theorem below.  We denote a collection of CDS by
$\hat S$, and write $D \models \hat S$ if every CDS of $D$ is
dominated by the corresponding CDS in $\hat S$.  Recall that
$\Delta F(i) \defeq F(i)-F(i-1)$, thus $\Delta \hat S$ represents the
DS associated to the CDS in $\hat S$.

\begin{thm}\label{thm:cdf:upper:bound}
Suppose $Q$ is a Berge-acyclic query, and let $\hat S$ be a set of cumulative degree sequences, one for each relation and each attribute. Then, the following is true,
\begin{align}
    \forall\,D\models \hat S \quad |Q(D)|\leq |Q(W(\Delta \hat S))| \label{eq:cdf:upper:bound}
\end{align}
\end{thm}

The proof of this new theorem is a non-trivial extension of
Theorem~\ref{th:dsb}, and uses some results
from~\cite{DBLP:journals/corr/abs-2201-04166} as well as new results;
we defer it to our online appendix \cite{github}.
% The intuition from this theorem can be seen by considering the
% example in Figure~\ref{fig:worst:case:instance}. Replacing the true
% CDS $F_{S.X}$ with an upper bound $\hat{F}_{S.X}$ is equivalent to
% relabeling values of $W(s_R)$ with lower values. For example,
% consider replacing the $2$ in the fourth row with a $1$. The
% resulting CDS is an upper bound of the true CDS, and the fourth row
% now joins with a higher frequency value in adjoining tables than it
% did before, increasing the join output size on the new worst-case
% instance and therefore the bound.
%
This justifies the following:

\begin{definition} \label{def:safe:compression} Let $R.V$ be a column
  with degree sequence $f_{R.V}$.  We say that $\hat f$ is {\em valid}
  for $f_{R.V}$ if (a) it is a degree sequence, meaning it is
  non-increasing, $\hat{f}(i-1) \geq \hat{f}(i)$, (b) its CDS
  $\hat F(i) \defeq \sum_{j \leq i}\hat{f}(i)$ dominates the CDS of $R.V$:
  $F_{R,V}(i) \defeq \sum_{j\leq i} f_{R.V}(i)) \leq \hat F(i)$,
  $\forall i$, (c) it preserves the cardinality
  $|R|= F_{R.V}(|\mathbb{D}_V|) = \hat F(|\mathbb{D}_V|)$.
\end{definition}

To summarize, SafeBound compresses every DS $f_{R.V}$ into a valid DS
$\hat f_{R.V}$.  Since $\hat f_{R.V}$ no longer dominates $f_{R.V}$,
we need to make some small adjustments to the way we compute
conditioned degree sequences in Sec.~\ref{subsec:choosing:statistics},
as follows. The max-degree sequence for the non-MCV values,
Eq.~\eqref{eq:max:eq}, will be computed over CDS rather than DS, in
other words we replace it with
$\hat F_{R.V|A} \defeq \max_\ell \hat F_{R.V|(A=a_\ell)}$: this,
improves the bound.  The min-degree sequence for a conjunction of
predicates in Eq.~\eqref{eq:min:conjunction} will also use the CDS, in
other words it becomes
$\hat F_{R.V|(A \wedge B \wedge \cdots)}(i) \defeq \min (\hat
F_{R.V|A}(i), \hat F_{R.V|B}(i), \ldots)$: this worsens the bound, but
this is necessary to ensure correctness.  All other computations
remain unchanged, with each $f$ replaced by $\hat f$. Next, we discuss how
to compute a good valid compression for a given degree sequence.

\subsection{The Compression Algorithm}\label{subsec:modeling:cdfs}

In this section, we describe the compression algorithm of a degree
sequence $f$ to $\hat f$.  Function approximations are defined by a
model class (e.g. polynomial, sinusoidal, etc), a loss function
(e.g. mean squared error), and an approximation algorithm
(e.g. gradient descent, convex hull, etc).  We have three
requirements: (1) the approximation of the CDS must be an upper bound
of the original CDS, i.e.  $F \leq \hat F$ (by
Th.~\ref{thm:cdf:upper:bound}), (2) the model class must be closed
under both multiplication with the derivative and composition, and (3)
the model class must be invertible.  The last two requirements are
needed by Algorithm~\ref{alg:FDSB}, which is described in the next section. Given these requirements, SafeBound uses piecewise linear
representation for $\hat F$, or, equivalently, piecewise constant for
$\hat f$:

\begin{defn}
  A function $f(x)$ is \textit{$k$-piecewise linear} over the domain
  $(m_0, m_{k}]$ if there exist tuples $\{(m_0,m_1, a_1x+b_1),\ldots,$
  \\$(m_{k-1},m_{k}, a_kx+b_k)\}$ such that $f(x)=a_ix +b_i\,\,\forall
  x\in(m_{i-1},m_i]$\\$\forall 1\leq i\leq k$. We call the values,
  $m_0, \ldots, m_k$, {\em dividers} and call each interval,
  $(m_{\ell-1}, m_\ell]$, a {\em segment}.  When $a_1=\cdots=a_k=0$,
  we say that $f$ is {\em piecewise constant}.
%   We denote the set
%   of $k$-piecewise linear functions is denoted $\mathcal F_k$.
\end{defn}

Every piecewise linear function with $k$ segments can be stored in
$O(k)$ space.  If $\hat f$ is piecewise constant, then
$\hat F(i) \defeq \sum_{j \leq i}\hat f(j)$ is piecewise linear, and,
conversely, if $\hat F$ is piecewise linear and continuous, then
$\hat f \defeq \Delta \hat F$ is piecewise constant.  SafeBound
compresses degree sequences as piecewise constant, or, equivalently
compresses cumulative degree sequences as piecewise linear functions;
the conversion from one to the other is done in time $O(k)$.

Degree sequences naturally compress very well.  If $R.V$ is a key,
then its degree sequence $f(1)=f(2)=\cdots=f(N)=1$ compresses
losslessly to a single segment, $k=1$.  Even if $R.V$ is not a key,
its degree sequence can still be compressed losslessly:

\begin{lmm}\label{th:lossless:compression}
  Let $f$ be the degree sequence of a column $R.V$, and suppose $R$
  has $N$ tuples.  Then $f$ compresses losslessly to a piecewise
  constant function with $k \leq \min(\sqrt{2N},f(1))$ segments.
\end{lmm}

\begin{proof}
  We have $f(1) \geq f(2) \geq \cdots \geq f(d)$, where $d$ is the
  number of distinct values in $R.V$.  Assume w.l.o.g. that $f(d)>0$
  (otherwise, decrease $d$).  Consider its natural dividers into $k$
  segments, $0=i_0 < i_1 < \ldots < i_k\defeq d$ such that $f(j)$ is
  constant for $i_{\ell-1} < j \leq i_{\ell}$ and
  $f(i_1) >\cdots > f(i_k)>0$.  Since $f$ takes integer values, it
  follows that $f(1)=f(i_1)\geq k$.  On the other hand,
  $|R| = \sum_i f(i) \geq f(i_1)+\cdots+f(i_k)\geq k+(k-1)+\cdots+1+0
  = k(k+1)/2 \geq k^2/2$, which implies $k \leq \sqrt{2N}$.  This
  proves $k \leq \min(\sqrt{2N},f(1))$.
\end{proof}

\begin{algorithm}[t]
  % \caption{Two-Pass Compression Algorithm}
\caption{ValidCompress}
\label{alg:modeling}
\begin{algorithmic}[1]
\REQUIRE $f_{R.V}, c$ \hfill // $f_{R.V}=$ exact DS, $c=$ accuracy parameter
\STATE $d = |\mathbb{D}_{R.V}|$  \hfill // the number of distinct values in $R.V$
\STATE $SJ = \sum_{i=1}^d(f_{R.V}(i)^2)$ \hfill // exact DSB of selfjoin
\STATE // initialize 1st segment $(m_0,m_1]$:
\STATE $k=1; \epsilon_1=0; m_0=m_1=0$; $a_1=f_{R.V}(1); b_1=0$
% \STATE $a_k = f_{R.V}(1)$ // The current segment's slope
% \STATE $b_k, m_k, \epsilon_k = 0$ //  The current segment's intercept, right edge, and error
\FOR{$i \in [1,\ldots, d]$} \label{alg:compress:for:loop}
\STATE $\epsilon_k = \epsilon_k + a_k^2\cdot(f_{R.V}(i)/a_k) -f_{R.V}(i)^2$
% \STATE // If the self-join error exceeds $c\cdot SJ$, start a new segment.
\IF{$\epsilon_k \geq c\cdot SJ$}
\STATE // DSB error too big?
\STATE $k = k+1; \epsilon_k = 0$  \hfill // start new segment $(m_{k-1},m_k]$ \label{alg:compress:new:segment}
\STATE $m_k = m_{k-1}; a_k = f_{R.V}(i); b_k = b_{k-1} + a_{k-1}(m_{k-1}-m_{k-2})$ \label{alg:compress:ak}
\ENDIF
\STATE $m_k = m_k + f_{R.V}(i)/a_k$ \hfill // extend  segment $(m_{k-1},m_k]$  \label{alg:modeling:width}
\ENDFOR  \label{alg:compress:for:loop:end}
\STATE $\hat{F}_{R.V} = \{\left(m_{l-1},m_{l}, a_l(x-m_{l-1}) + b_l\right)| l=1,k\} \cup \{(m_k, d, |R|)\}$ \label{alg:compress:return}
\RETURN $\hat f_{R.V} \defeq \Delta \hat{F}_{R.V}$ \hfill // $k+1$ segments
\end{algorithmic}
\end{algorithm}

SafeBound does not rely on the natural compression, but instead uses a
more aggressive, lossy compression, with a much smaller number of
segments than given by the lemma.  Algorithm~\ref{alg:modeling},
called~\textbf{ValidCompress}, takes as input a degree sequence
$f_{R.V}$, and an accuracy parameter $c > 0$, and computes a valid
compression using the following heuristic: if
$SJ = \sum_{i=1}^{|\mathbb{D}_V|}f_{R.V}(i)^2$ is the exact Degree
Sequence Bound of the self-join on the column $R.V$, then the
algorithm ensures that no segment increases the DSB by more than
$c \cdot SJ$.

We describe the algorithm and prove its correctness.  It iterates
through the degree sequence $f_{R.V}(i)$, $i=1,2,3,\ldots$ and builds
the segments of $\hat F_{R.V}$ one by one:
$(m_0\defeq 0, m_1], (m_1,m_2], \ldots, (m_{k-1},m_k]$.  Initially,
$k=1$, the first segment is empty $(0,0]$, and the initial slope is
$a_1=f_{R.V}(1)$.  The \texttt{for}-loop in
lines~\ref{alg:compress:for:loop}-\ref{alg:compress:for:loop:end}
iterates over each rank $f_{R.V}(i)$ and does one of two things: it
either extends the current segment $(m_{k-1},m_k]$ by increasing $m_k$
(line~\ref{alg:modeling:width}), or it increases $k$ and starts a new
empty segment (line~\ref{alg:compress:new:segment}), which is also
immediately extended in line~\ref{alg:modeling:width}.  The choice
between these actions is dictated by our heuristics: ensure that
each segment contributes at most $c\cdot SJ$ to the DSB.  We prove
that, regardless of the heuristic, the algorithm always computes
a valid compression, by checking conditions (a), (b), (c) in
Def.~\ref{def:safe:compression}.  The following invariant holds at the
beginning of each iteration of the \texttt{for}-loop
(line~\ref{alg:compress:for:loop}): if $\hat F_{R.V}$ denotes the
current piecewise linear function, defined on $(0,m_k]$, then:
\begin{align*}
\hat F_{R.V}(m_k) = & F_{R.V}(i)
\end{align*}
This follows by induction on $i$.  Before the first iteration, $i=0$,
$m_1=0$ and $\hat F_{R.V}(0)=F_{R.V}(0)=0$.  Consider the inductive
step, from $i-1$ to $i$.  On one had, the value of $F_{R.V}(i)$ grows
by $f_{R.V}(i)$; on the other hand, $m_k$ increases in
line~\ref{alg:modeling:width} and its current slope is $a_k$, hence
the value $\hat F_{R.V}(m_k)$ will grow by exactly
$a_k \cdot (f_{R.V}(i)/a_k) = f_{R.V}(i)$, proving the invariant.  In
particular, this implies that $\hat F_{R.V}(m_k)$ is always
$\leq F_{R.V}(d)=|R|$, where $d=|\mathbb{D}_V|$: it justifies adding
the a constant segment $(m_k,d,|R|)$ (line~\ref{alg:compress:return}),
and proves condition (c): cardinality is preserved.  Since the slopes
$a_k$ defined in Line~\ref{alg:compress:ak} are decreasing, condition
(a) holds: $\Delta \hat F_{R.V}$ is decreasing.  Finally, condition
(b), $\hat{F}_{R.V}(i) \geq F_{R.V}(i)$, follows from the fact that
during the \texttt{for}-loop $m_k \leq i$, since $i$ always grows by
1, while $m_k$ grows by $f_{R.V}(i)/a_k \leq 1$.  Using the invariant
and the fact that $\hat F_{R.V}$ is monotonically increasing, we
obtain $\hat{F}_{R.V}(i) \geq \hat F_{R.V}(m_k) = F_{R.V}(i)$, proving
(b).
In summary:
\begin{theorem}
  Algorithm~\ref{alg:modeling} computes a valid compression of
  $\hat f_{R.V}$ of $f_{R.V}$, with $k+1$ segments and a relative
  self-join error $\leq c\cdot k$.
\end{theorem}
Our algorithm is loosely inspired by approximate convex hull
algorithms such as the one used in~\cite{ferragina2020pgm}, and it is similarly linear in time and space with respect to the degree sequence length. Further, calculating a DS from a column of length $n$ requires $O(n\log(n))$ time and $O(n)$ space. In our implementation, we typically choose $c=.01$ which results in
$k = 20-30$ segments for compressing the DS of a foreign key and $<10$
segments for the DS conditioned on an element of the MCV
list. Further, if the join column is a key, then it always compresses
to a single segment.

{\bf Discussion}  We briefly justify our choice of heuristics over other possible choices. One choice would be to minimize the absolute distance between the true CDS and the approximation, $\sum_{i=1}^{|\mathbb{D}_V|}|F_{R.V}(i)-\hat{F}_{R.V}(i)|$. However, this distance would treat errors on high frequency and low frequency values as equally undesirable when the high frequency values actually have a much larger impact on the final bound. This is due to high frequency values joining with high frequency values in the worst-case instance. Alternatively, one could choose some specific weighted distance to use for modeling all columns, $\sum_{i=1}^{|\mathbb{D}_V|}w_i|F_{R.V}(i)-\hat{F}_{R.V}(i)|$. However, because that optimal weighting will depend on the adjoining tables, choosing a single weighting for all columns assumes that they will all have similarly distributed adjoining tables. For instance, this would imply that a column containing country IDs will join with the same columns as one that contains employee IDs. Our choice of the self-join error metric amounts to assuming that tables will join with similarly skewed tables. Future work may consider the skewness of adjoining tables in the database schema or a sample workload to create a more accurate metric.

\subsection{Fast Computation of the Upper Bound} \label{subsec:fdsb}
We finally turn to the online phase of SafeBound: given a query and
the collection of compressed degree sequences, use the statistics
$\hat s$ to compute the upper bound $|Q(W(\hat s))|$
(Equation~\eqref{eq:cdf:upper:bound}).  Throughout this section, we
assume that $\hat s$ are valid compressions (see
Def.~\ref{def:safe:compression}) and represented by piecewise constant
functions; equivalently, their CDS $\hat S$ are piecewise linear
functions.  We assume that all predicates have been applied to the
base tables, and $\hat s$ includes all conditional degree sequences
needed for the predicates, as discussed in
Sec.~\ref{subsec:choosing:statistics}; in other words, we will assume
w.l.o.g.  that $Q$ consists only of joins, and no predicates.
Referring to the running Example~\ref{ex:running:example}, we assume
that the query is $R'(X) \wedge S'(X,Y) \wedge T(Y)$, where the degree
sequence of $R'.X$ is
$\min(\hat F_{R.X|(A<5)}(i), \hat F_{R.X|(B=2)}(i))$ (see the
discussion at the end of Sec.~\ref{subsec:cdsb}) and the DS for $S'.X$
and $S'.Y$ are given by conditioning on the predicate
$\texttt{LIKE}\,\,'\%Abdul\%'$.

The naive computation requires materializing the worst case instance,
$W$, and is totally impractical, since $W$ is at least as large as the
database instance, regardless of how well we compress the statistics
$\hat s$.  Instead, SafeBound implements a more efficient algorithm,
adapted from~\cite{DBLP:journals/corr/abs-2201-04166}, which avoids
materializing $W$, but instead computes the bound directly, in time
that depends only on the total size of all compressed degree
sequences.

The starting observation is that $Q$ is acyclic, and can be computed
bottom-up, where at each tree node we join the current relation with
its children and project out all attributes except the unique
attribute needed by its parent.  We write this plan as an alternation
between two kinds of operations, which we call $\alpha$ and $\beta$
steps:
\begin{align}
\alpha: &&A(X) &= B_1(X) \wedge \ldots \wedge B_m(X) \label{eq:alpha} \\
\beta:  &&B(X_0) &= R(X_0, X_1,\ldots,X_k) \wedge A_1(X_1)\wedge \ldots \wedge A_k(X_k) \label{eq:beta}
\end{align}
An $\alpha$-step intersects unary relations, while a $\beta$-step is a
star-join followed by a projection on a single variable.  Recall that
all our queries have bag semantics, so this projection does not reduce
the cardinality.  The cardinality of $Q$ is the cardinality of the
last unary relation, corresponding to the root of the tree.

\begin{example} \label{ex:acyclic}
  We briefly illustrate the $\alpha,\beta$ steps for the query $Q$:
    \begin{align*}
        &R(X,Y,Z) \wedge S(Y) \wedge K(Z) \wedge T(Z,V,W)\wedge  M(V)\wedge N(V)\wedge P(W)
    \end{align*}
    We only show the attributes used in joins: e.g.  relation $S$ may
    have attributes $S(Y,A,B,C,D,\ldots)$ but we only show the join
    attribute $Y$. Fig.~\ref{fig:acyclic} shows a tree decomposition
    for $Q$, and a plan consisting of two $\alpha$ and two $\beta$
    steps.  The cardinality of the original query, $Q$, is the
    cardinality of the last unary relation, $B'$.
  \begin{figure}[h]
    \centering
    \begin{minipage}[c]{0.4\linewidth}
      \includegraphics[width=\linewidth]{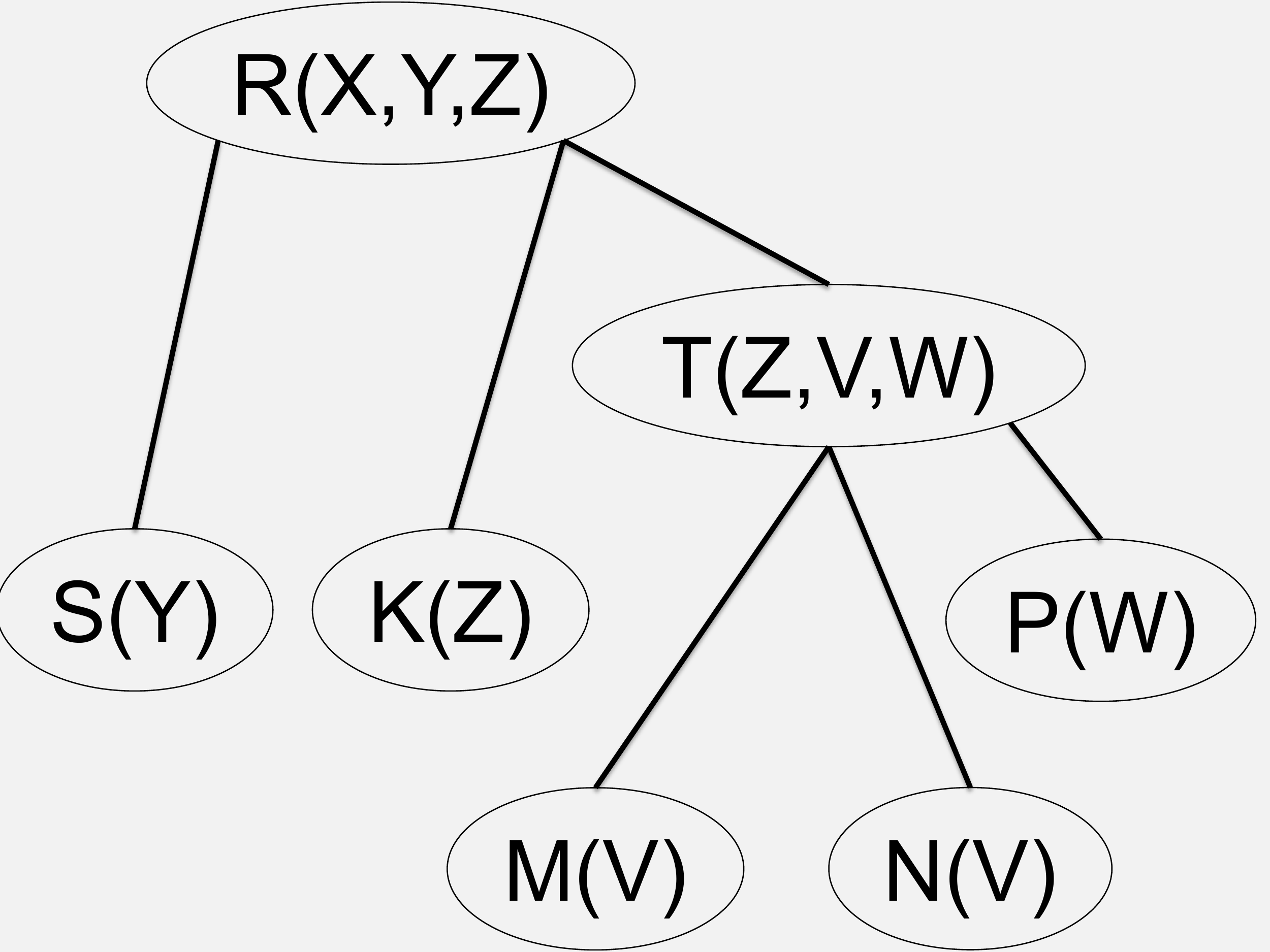}
    \end{minipage}
    \begin{minipage}[c]{0.4\linewidth} {\small
    \begin{align*}
        A(V) &= M(V)\wedge N(V)\\
        B(Z) &= T(Z,V,W)\wedge A(V)\wedge P(W)\\
        A'(Z) &= K(Z) \wedge B(Z) \\
        B'(X) &= R(X,Y,Z)\wedge S(Y)\wedge A'(Z)
    \end{align*}
  }
\end{minipage}
    \caption{A tree decomposition for $Q$ in Example~\ref{ex:acyclic}}
    \label{fig:acyclic}
  \end{figure}
\end{example}

Algorithm~\ref{alg:FDSB} evaluates $Q$ on the worst case instance
$W(\Delta \hat S)$ without materializing $W$.  Since the algorithm
uses piecewise linear CDS, the computation may lead to rounding errors
causing a slight over-approximation of the upper bound in
Eq.~\eqref{eq:cdf:upper:bound}.  For that reason, the bound computed
by the algorithm is called the {\em Functional Degree Sequence Bound},
FDSB. The key observation in the algorithm is that every unary
relation $A(X)$ or $B(X_0)$ is also piecewise constant; for that
reason we call them in the algorithm $\hat f_{A.X_0}(i)$ and
$\hat f_{B.X}(i)$ respectively.  To see this, consider an
$\alpha$-step: $A(X) = B_1(X) \wedge B_2(X) \wedge \cdots$ If each
$B_\ell$ contains the values $1,2,3,\ldots$, then so will $A$, and the
multiplicity of the value $i$ is the product of the multiplicities in
the $B_\ell$'s: this justifies line~\ref{alg:line:alpha}:
$\hat f_{A.X}(i) = \hat f_{B_1.X}(i) * \hat f_{B_2.X}(i) * \cdots$ The
product of piecewise constant functions is still piecewise constant,
with a number of segments equal to the sum of all segments of the
summands.  Consider now a $\beta$-step, Eq.~\eqref{eq:beta}. The
multiplicity of $i$ in the output $B$ is the product
$\hat f_{R.X_0}(i) * \hat f_{A_1.X_1}(i_1) * \hat f_{A_2.X_2}(i_2) *
\cdots$, where the ranks $i_1, i_2, \ldots$ need to be looked up using
the expression $i_\ell = \hat{F}_{R.X_\ell}^{-1}(\hat{F}_{R.X_0}(i))$;
this justifies line~\ref{alg:line:invert}.  Here, too, we observe that
$i_\ell$ is given by a piecewise linear function in $i$, while
$\hat f_{A_\ell.X_\ell}(i_\ell)$ is piecewise constant.  In our
implementation, both lines~\ref{alg:line:alpha}
and~\ref{alg:line:invert} of the algorithm compute a {\em
  representation} of the resulting piecewise constant functions. This
representation consists of three vectors; slopes, intercepts, and
right interval edges. Because of this representation, multiplication
of functions only requires a single pass over these arrays while
computing the inverse term,
$\hat{F}_{R.X_\ell}^{-1}(\hat{F}_{R.X_0}(i))$, involves performing a
binary search for each of the interior function's segments. Finally, the algorithm returns the cardinality of the last unary relation at the root.  The following theorem from~\cite{DBLP:journals/corr/abs-2201-04166} also applies to our setting:

\begin{algorithm}[t]
  \caption{Algorithm for $\text{FDSB}(Q,\hat S)$.}
\label{alg:FDSB}
\begin{algorithmic}[1]
\REQUIRE A query plan for $Q$ consisting of $\alpha,\beta$ steps.
\FOR{\mbox{each step}}{
    \IF{$\alpha$-step}{
        \STATE // $A(X) = B_1(X)\wedge \ldots \wedge B_m(X)$
        \STATE $\hat{f}_{A.X}(i) =\prod_{\ell=1}^m\hat{f}_{B_\ell.X}(i)$  \label{alg:line:alpha}
        \STATE // Note: $\hat{f}_{A.X}$ is a piecewise constant function
    }\ENDIF
    \IF{$\beta$-step}{
        \STATE // $B(X_0)=R(X_0, X_1,\ldots,X_k)\wedge A_1(X_1)\wedge \ldots\wedge A_k(X_k)$
        \STATE $\hat{f}_{B.X_0}(i) =  \hat{f}_{R.X_0}(i)\prod_{\ell=1}^k\hat{f}_{A_\ell.X_\ell}\left(\hat{F}_{R.X_\ell}^{-1}(\hat{F}_{R.X_0}(i))\right)$ \label{alg:line:invert}
        \STATE // Note: $\hat{f}_{B.X_0}$ is a piecewise constant function
    }\ENDIF
 }
\ENDFOR
\RETURN  $\sum_i \hat{f}_{B_{\text{root}}.X}(i)$
\label{alg:FDSB:return}
\end{algorithmic}
\end{algorithm}

\begin{thm}\label{th:fdsb}
  Let $\hat S$ consists of piecewise linear CDS, and let $K$ be the
  total number of segments occurring in all CDS.  Then,
\begin{enumerate}
\item \label{item:fdsb:1}
  $|Q(W(\Delta \hat S))| \leq \text{FDSB}(Q, \hat S)$.
\item \label{item:fdsb:2} $\text{FDSB}(Q,\hat S)$ can be computed in
  time $O(K\log(K))$.
\end{enumerate}
\end{thm}

The first item asserts that FDSB is a correct upper bound.  The second
item shows that we can compute it in log-linear time relative to the
size of the compressed representations. We
usually have $10-30$ segments per degree sequence, thus a query with 5
joins has at most $K=300$. This results in very fast inference as
shown in Section~\ref{sec:evaluation}.

\subsection{Discussion: More Complex Join Graphs}
\label{subsec:cyclic-multi-key}
\noindent\textbf{Cyclic Queries:} If $Q$ is a cyclic query, then we compute its upper bound as the minimum of the DSBs of all its spanning trees (which are acyclic).  For example, if $Q$ is $R(X,Y)S(Y,Z)T(Z,X)$, there are three spanning trees $R(X,Y)S(Y,Z)$, and $R(X,Y)T(Z,X)$, and $S(Y,Z)T(Z,X)$, and we return the minimum of their DSB.

\noindent\textbf{Multi-Column Joins:} We provide two methods of handling multi-column joins. (1) Treat multiple columns as a single column at construction time and calculate its CDS just like for a single column. (2) Alternatively, we note that the CDS for any single column is an upper bound on the CDS of a set of columns.  Therefore, we can always upper bound the CDS of the set of columns by taking the minimum of the CDS of each of its columns. 

\noindent\textbf{Undeclared Join Columns:} We provide a fallback mechanism to handle queries with joins on columns that are not in the declared join column set. During offline construction, we keep a compressed CDS for every column in the relation without conditioning on any predicates. This incurs minimal overhead as it only calculates one CDS per column. At query-time, we calculate a CDS for any declared join column and use it to derive an upper bound on the filtered relation's cardinality. We then truncate the undeclared join column's CDS to match this single-table cardinality bound and use this to approximate the filtered CDS of the undeclared join column. This allows us to adapt the bound to the lower cardinality induced by any predicates on the relation. However, the resulting CDS is likely overly skewed because it assumes that the predicate retains precisely the high degree values.

%% file: 4-optimizations.tex
\section{Optimizations}
\label{sec:optimizations}
Here, we present optimizations that reduce the space consumption and increase the accuracy of the vanilla design.

\subsection{Compressing a Set of CDS}\label{subsec:grouping:CDS:sets}
To reduce the memory overhead of storing a CDS set for every bin, value, and N-gram, we go beyond compressing a single CDS and consider compressing groups of CDS sets together. Specifically, we divide the CDS sets into groups of "similar" functions and replace each group with the point-wise maximum of those sets.

To motivate this, we first present a breakdown of the memory cost of SafeBound's statistics. Let $b$ be the granularity, e.g. the number of buckets in the histogram, and let $k$ be the number of segments in each CDS. Further, let $|V_J|$ and $|V_F|$ be the number of join and filter columns, respectively. The memory footprint without performing any grouping compression is $O(b|V_F| + bk|V_F||V_j|)$ where the first term is the bucket bounds or values in the MCV list while the second term is the cost of storing the CDS sets. Now, consider dividing the CDS sets into $M$ groups and storing just the maximum over each group. The memory footprint then becomes $O(b|V_F| + Mk|V_F||V_j|)$ which allows us to decouple the granularity of our statistics, $b$, from the accuracy of our approximations, $Mk$. This is crucial for workloads which feature highly selective predicates because it allow us to keep more fine-grained histogram buckets, MCV lists, and N-grams. 

As an example, consider a range predicate $.1\leq R.A<.2$. We may only have the memory to store buckets of width $1$ if we store every bucket's CDS exactly. However, if we cluster and compress our CDS sets with an average cluster size of $10$, we may able to have buckets of width $.1$ which encapsulate the query much tighter while only incurring a $40\%$ relative approximation error. The relative approximation error in this case is far outweighed by the improved granularity of our statistics.

\noindent\textbf{Choosing a Distance Metric:} The first step to clustering is choosing a distance metric for the problem. The perfect distance metric for this problem is the average error incurred on the workload when the two functions are replaced with their maximum. However, we don't have access to the workload when clustering, so we instead use the same assumption that we used in Sec. \ref{subsec:modeling:cdfs}, that the workload consists of self-joins. Therefore, our distance metric becomes the self-join error, i.e.
\begin{align*}
    d(F_1,F_2)=\frac{\sum_{i=1}^{|\mathbb{D}_V|}\Delta\max(F_1(i),F_2(i))^2}{\sum_{i=1}^{|\mathbb{D}_V|}f_1(i)^2} +\frac{\sum_{i=1}^{|\mathbb{D}_V|}\Delta\max(F_1(i),F_2(i)^2}{\sum_{i=1}^{|\mathbb{D}_V|}f_2(i)^2}
\end{align*}.

\noindent\textbf{Choosing a Clustering Algorithm:}
Given this distance metric, we need to choose a clustering algorithm, and we choose \textit{complete-linkage clustering} ~\cite{mullner2013fastcluster}. This method of hierarchical clustering defines the distance between clusters as the maximum distance between points in each cluster. As opposed to other clustering methods such as single-linkage clustering, it produces tighter clusters and avoids long "chain" clusters which contain highly dissimilar points. This results in clusters of functions which are well approximated by their point-wise maximums.

\subsection{Pre-Computing Primary Key Joins}
\label{subsec:primary:key}
\textbf{Predicates Induce Cross-Join Correlation:} As described in Section \ref{subsec:dsb}, the FDSB makes worst-case assumptions about the correlation of columns in joining tables. This assumption is fundamental to computing an upper bound. However, particularly in the presence of predicates, these assumptions may not hold, causing SafeBound to overestimate the query size. 

For example, consider the tables \texttt{MovieKeywords} and \texttt{Keywords} from the JOB Benchmark. The former is a fact table with two foreign key columns, \texttt{MovieId} and \texttt{KeywordId}, that associate movies with keywords. The latter is a much smaller dimension table with a primary key column \texttt{KeywordId} and a filter column \texttt{Keyword}, which provides human-readable descriptions of these keywords, e.g. 'character-name-in-title' or 'pg-13'. A natural query would join them with an equality predicate on the \texttt{Keyword} column to find movies with a particular keyword. A naive version of SafeBound would assume that the selected keyword corresponds to the most frequent value of \texttt{KeywordId} in the \texttt{MovieKeywords} table. If the queried keyword actually occurs infrequently in \texttt{MovieKeywords}, this could introduce a massive error in the final estimation. 

\textbf{Handling Predicate-Induced Correlation:} To avoid this issue, SafeBound pre-computes PK-FK joins and stores statistics about the filter columns of the PK relations. In our example, this would mean joining \texttt{MovieKeywords} and \texttt{Keywords} then generating statistics on the resulting \texttt{keyword} column in \texttt{MovieKeywords}. When an equality predicate is applied to the \texttt{keyword} column on the \texttt{Keywords} table, SafeBound applies this predicate to the \texttt{MovieKeywords} table as well, allowing it to directly estimate the CDS set given the predicate without resorting to worst-case assumptions.

Fortunately, the PK-FK join size is bounded by the size of the FK table, so this pre-computation is tractable. While this does not capture all correlations, it does enable accurate estimation for the ubiquitous fact/dimension table design where predicates are applied to dimension tables then propagated to fact tables via PK-FK joins.

\subsection{Bloom Filters}
\label{subsec:bloom:filters}
An important source of overhead in SafeBound's data structures are the most common values lists (MCV lists) that it keeps for handling equality predicates. Because values can have an unbounded size, storing a naive MCV list can result in significant memory and lookup overhead. To avoid this overhead, we instead represent our MCV lists as a set of \textit{Bloom filters}. A Bloom filter is an approximate data structure, which answers the question "is $x$ an element of the set $S$?" while allowing some false positives and no false negatives. In exchange for approximation, Bloom filters provide a compressed memory footprint ($\approx 12$ bits/value) and fast, constant lookup\footnote{There are many variants on Bloom filters which would work equally well here, e.g. Cuckoo filters, quotient filters, and XOR filters.}.

Because Bloom filters only return a positive/negative and SafeBound needs to connect values to their CDS group, we can't represent the whole MCV list in one filter. Instead, we allocate a filter for each CDS group and insert all values whose CDSs are in that group into its filter. At query time, SafeBound then checks for membership in every group's filter and takes the maximum over all CDS sets whose filter return positive.

%% file: 5-evaluation.tex
\begin{figure*}[t]
    \centering
    \subcaptionbox{Total workload runtimes relative to runtimes achieved with perfect cardinality estimates. SafeBound results in nearly optimal overall runtimes.
    \label{subfig:runtime-graphs}}[.3\textwidth]{
    \includegraphics[height=120pt, width=165pt]{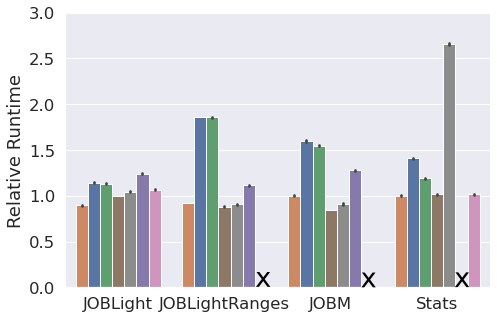}
    }
    \hfill
    \subcaptionbox{SafeBound achieves \textbf{3x-500x} faster median planning time than PessEst and ML-based methods across all benchmarks.
    \label{fig:planning-time}}[.3\textwidth]{
    \includegraphics[height=120pt, width=165pt]{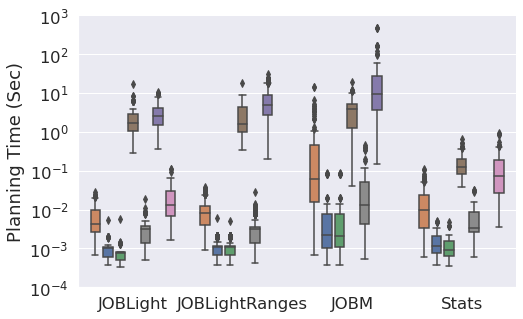}
    }
    \hfill
    \subcaptionbox{Relative errors for cardinality estimates. SafeBound's bounds have similar errors to traditional estimates while never underestimating cardinalities.
    \label{subfig:relative-error}}[.36\textwidth]{
    \includegraphics[height=116pt, width=220pt]{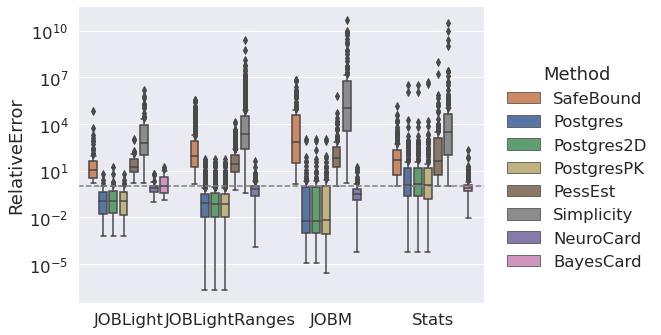}
    }
    \caption{Workload runtimes, planning time, and estimation error}
    \label{fig:runtime-graphs}
\end{figure*}

\section{Evaluation} \label{sec:evaluation}

%The main takeaways are: \magda{To save space, we can remove the main takeaways from here. We can have a shorter sentence at the end of the introduction. Alternatively, if we keep the main takeaways, I would point to the appropriate figure. }
%\begin{itemize}
%    \item Using SafeBound's estimates Postgres produces plans that achieve workload runtime \textbf{within $10\%$} of those generated using the true cardinality across all 4 benchmarks.
%    \item SafeBound uses up to \textbf{6.8x} less memory than competing ML methods, performs inference up to \textbf{500x} faster, and is constructed up to \textbf{19x} faster.
%    \item SafeBound produces valid upper bounds on the true cardinality of queries. \magda{Do we show this last point in the eval?}
%\end{itemize}

In this section we present an empirical evaluation of SafeBound.  We
addressed the following questions.  How well does SafeBound perform in
end-to-end workloads (Sec.~\ref{subsec:end:to:end})? How does its
memory footprint and inference time compare to existing methods
(Sec.~\ref{subsec:cost})?  How does SafeBound affect DBMS robustness, e.g. performance
regressions when new indices are added (Sec.~\ref{subsec:regression})?
We also conducted several micro-benchmarks in
Sec.~\ref{subsec:micro:benchmarks}, and explored how SafeBound scaled in Sec. \ref{subsec:scalability}.

\textbf{Metrics} We used the following metrics in our evaluation.  (1)
Plan Quality (Workload Runtime): Following recent work on benchmarking
cardinality
estimators~\cite{DBLP:journals/pvldb/HanWWZYTZCQPQZL21}
we measure the end-to-end runtime of a query workload in Postgres
where we injected alternate cardinality estimators into the optimizer.
We run each workload and method five times from a cold cache and
present the average relative to the baseline of inserting the true
cardinality estimates.  (2) Memory Footprint: We compare the size of
the stored statistics file on disk, and for Postgres we calculate the
size of the \texttt{pg\_statistic} and
\texttt{pg\_statistic\_extended} catalog tables. We do not report
memory statistics for PessEst as it does not pre-compute statistics.
(3) Planning Time: We further consider the planning time for each
method. This includes the inference time required to get estimates for
every sub-query as well as Postgres' optimization time given injected
estimates. (4) Relative Error: Lastly, we present the relative error
of each method as
$\text{Error} = (\text{Estimate}/\text{True Cardinality})$. We prefer
this metric to $q$-error as it retains information about whether a
method overestimates or underestimates.

\textbf{Datasets} We use two datasets, \textit{IMDB} and
\textit{Stats}. For \textit{IMDB} we consider three different query
workloads from previous
work~\cite{DBLP:journals/pvldb/YangKLLDCS20,DBLP:conf/cidr/KipfKRLBK19,DBLP:journals/pvldb/HilprechtSKMKB20}\footnote{We do not use the original JOB benchmark because it contains negation predicates which are not supported by SafeBound, NeuroCard, and BayesCard.}:
\textit{JOB-Light} consists of $70$ queries on a subset of $6$ tables
in IMDB with $2-5$ PK-FK joins and $1-4$ filter
predicates on numeric columns.  \textit{JOB-LightRanges} operates on
the same table subset as JOB-Light, but it has $1000$ queries,
includes additional columns, and predicates over
string columns.  And \textit{JOB-M} is a modified version of the
original JOB benchmark; it is the most complex benchmark considered,
with 113 queries over 16 tables, and includes significantly more
complicated expressions such as \texttt{IN} and \texttt{LIKE}
predicates.  The \textit{Stats} dataset is built over a Statistics
StackOverflow, and consists of a workload with $146$ queries spanning
$8$ tables. While restricted to numeric columns, it has $2-16$
predicates and joins $2-8$ tables per query making it is considered to
be a challenging benchmark for cardinality
estimation~\cite{DBLP:journals/pvldb/HanWWZYTZCQPQZL21}. It has a complicated schema with cyclic primary key/foreign
key relationships.

\textbf{Compared Systems} We compared SafeBound against the following
systems.  (1) Postgres: As a baseline, we compare against the built-in
cardinality estimator for Postgres v13. This system uses System-R
style estimation combined with years of tuning and carefully chosen
magic constants. It stores 1D histograms, most common value lists, and
distinct counts for each attribute in a relation.  (2) Postgres 2D:
We make use of Postgres' extended statistics, which allows the
user to keep statistics on pairs of columns. We instruct the system to
store statistics for every pair of filter columns.
% \revc{(3) Postgres PK: We provide Postgres with the PK-FK propagation optimization by materializing joins of tables with their PK join partners and allowing it to keep statistics on the additional filter columns. We then alter the queries to apply filter predicates on the PK tables directly to the FK tables and report on the accuracy. For example, consider PK-FK relationship $R(X,K)S(K, Y)$ where $R$ is a FK table and $S$ is a PK table. We pre-compute the join result $R'(X,K,Y)$ and alter predicates to reference $R'.Y$ rather than $S.Y$.}
(3) Postgres PK: SafeBound precomputes statistics on key,
  foreign-key joins, and so do BayesCard, and NeuroCard; PessEst
  computes PK-FK joins at query time when
  needed~\cite[Sec.3.3]{DBLP:conf/sigmod/CaiBS19}.  To understand the
  effect of these computations, we measured how much such
  precomputations could help Postgres.  We pre-computed and
  materialized the PK-FK joins, replaced the FK tables with this join
  (extending them with additional columns from the PK tables), and
  computed statistics on these tables.  We also adjusted the queries
  accordingly.  For example, consider the query
  $Q(X, A, B) = R(X,A)S(\underline{A}, B, Y)T(\underline{B}, Z) \wedge
  S.Y<10 \wedge T.Z > 5$ where $S.A$, $T.B$ are PKs.  We calculate the
  PK-FK join results $R'(X, A, Y')$ and $S'(A, B, Y, Z')$ and adjust
  the query to
  $Q'(X, A, B) = R'(X, A, Y') S'(\underline{A}, B, Y, Z')
  T(\underline{B}, Z) \wedge R'.Y'<10 \wedge S'.Z' > 5$.  We call this
  modifed system PostgresPK.  Notice that this mirrors our method by
  propagating statistics across PK-FK joins, without modifying the
  query's join graph.

(4) BayesCard: This is an ML method that uses ensembles of Bayesian Networks trained 
on subsets of the join schema to produce cardinality estimates~\cite{wu2020bayescard}.
Recent work has shown that it matches previous ML methods in accuracy while being
faster and more compact~\cite{DBLP:journals/pvldb/HanWWZYTZCQPQZL21}. (5) NeuroCard:
this is an ML method that builds an autoregressive model over a sample
of the full outer join of the schema~\cite{yang2020neurocard}. (6)
PessEst: The main prior work on cardinality
bounding~\cite{DBLP:conf/sigmod/CaiBS19}. It refines a subset of the
Polymatroid Bound using a hash partitioning scheme; we use $4096$ hash
partitions. However, this method requires scans of the base table to
estimate queries with predicates.
(7) Simplicity: a cardinality estimator which uses single-table
  cardinalities and max degrees of join
  columns~\cite{DBLP:conf/cidr/HertzschuchHHL21}. In order to improve
  the max degree in the presence of predicates, Simplicity relies on
  samples~\cite{DBLP:conf/cidr/HertzschuchHHL21} or on estimates
  derived from Postgres, which are no longer leading to {\em
    guaranteed} upper bounds.  In the original prototype and our
  implementation, the single-table estimates are derived from Postgres
  although more complicated sampling mechanisms are proposed in the
  paper. Similarly, we do not consider their greedy join ordering
  algorithm, instead focusing solely on the cardinality estimator.

\textbf{Experimental Setup} In our experiments, we use an instance of
Postgres v13 and input cardinality estimates using the
\texttt{pg\_hint\_plan} extension. We adjust the default settings of
Postgres per the recommendation of
~\cite{DBLP:journals/vldb/LeisRGMBKN18} setting the shared memory to
4GB, worker memory to 2GB, implicit OS cache size to 32 GB, and max
parallel workers to 6. Additionally, we enable indices on primary and
foreign keys. We run all experiments on an AWS EC2
instance (m5.8xlarge) with 32 vCPUs and 128 GB of memory.

\begin{figure*}[t]
    \centering
    \includegraphics[width=1\textwidth, height=120pt]{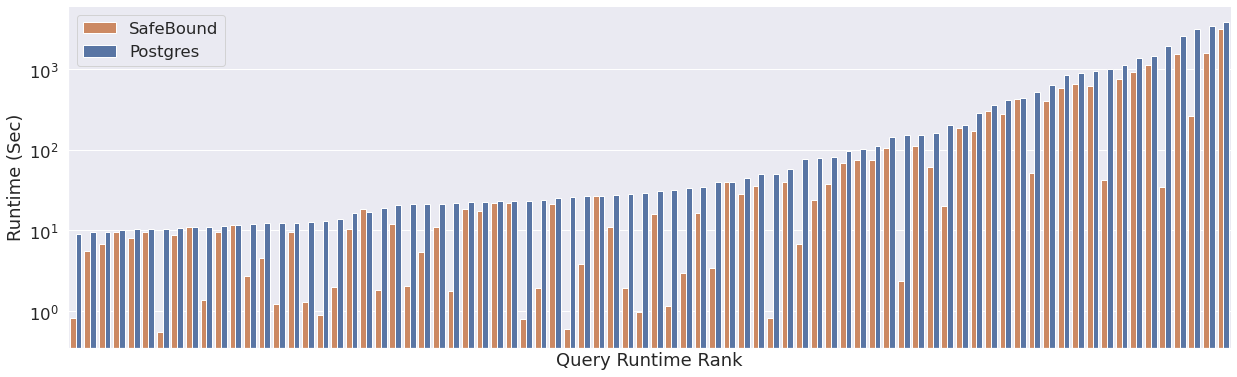}
    \caption{The runtime of the 80 longest-running queries across all benchmarks. Estimates from SafeBound result in faster query execution on the most expensive queries with a p05/p25/p50/p75/p95 quantile speedup of 1.01x/1.3x/1.7x/10.1x/30.3x.}
    \label{fig:Top-80-Runtime}
\end{figure*}
\begin{figure}[t]
    \includegraphics[width=.46\textwidth, height=125pt]{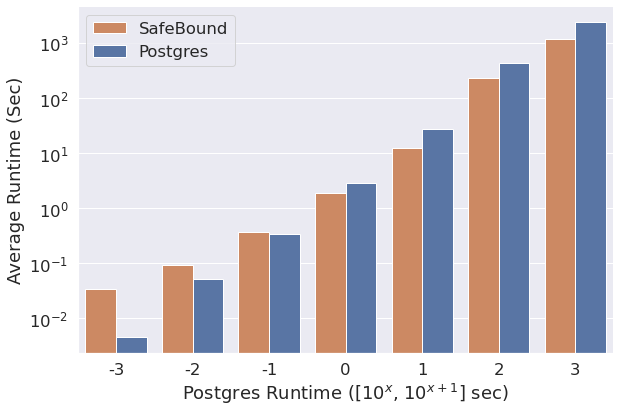}
    \caption{The average runtime of queries binned by their runtime using Postgres' estimates. SafeBound outperforms Postgres' estimates for queries with runtime over one second.}
    \label{fig:Join-Size-Vs-Runtime}
\end{figure}

\subsection{End-to-End Performance}
\label{subsec:end:to:end}
\noindent\textbf{Nearly optimal workload runtimes across all
  benchmarks.} We show the workload runtimes across a variety of
benchmarks and cardinality estimation methods in Figure
\ref{subfig:runtime-graphs}. Across all four benchmarks, plans
generated using SafeBound's estimates achieve workload runtimes
equivalent to those generated with the true cardinalities. As found in
previous work, using true cardinality does not always lead to optimal
plans due to imperfect cost modeling
\cite{DBLP:journals/pvldb/HanWWZYTZCQPQZL21}. SafeBound achieves
$20\%-85\%$ lower runtimes than Postgres on all
benchmarks. BayesCard and PessEst perform similarly to SafeBound on
all benchmarks while NeuroCard has $20-30\%$ worse performance on both
JOBLight and JOBM. Bayescard does not support the string predicates of JOB-LightRanges or
JOB-M, and NeuroCard does not support the cyclic schema of the Stats benchmark. All pessimistic systems achieve good performance on the JOB benchmarks, pointing to the utility of even fairly loose cardinality bounds. However, Simplicity results in a poor join ordering for query 132 of the Stats benchmark, resulting in a 1500x slowdown.

%To rule out the possibility that SafeBound's performance stems from extra cores, we repeated these experiments without parallelism and include the results in our online appendix\cite{github}. For the JOBLight, JOBLightRanges, and STATS benchmarks, the findings are similar to Figure~\ref{fig:runtime-graphs} . For JOB-M, SafeBound, Neurocard, and Postgres(2D) see comparable performance ($\approx 1.2$ relative runtime).

\noindent\textbf{Efficient plans for the queries that matter.} To provide context for SafeBound's performance, we examine the runtime of the longest running queries in Figure \ref{fig:Top-80-Runtime}. The queries that make up the bulk of the runtime can often be sped up significantly (up to \textbf{$60x$}) by using SafeBound instead of Postgres for cardinality estimates. 

Figure \ref{fig:Join-Size-Vs-Runtime} buckets the queries across all workloads by runtime and shows the average runtime using SafeBound and Postgres's estimates. Here, we can see SafeBound generally achieves a significant speedup for queries that take over a second. We see these speedups because cardinality bounds encourage the query optimizer to make conservative decisions (e.g. choosing hash joins over nested loop joins) which tend to be the correct decisions for queries with long runtimes or large outputs. For the fastest queries, SafeBound often results in slower execution as it discourages the optimizer from choosing high-risk high-reward plans.

\noindent\textbf{Estimation Errors.}
In Figure~\ref{subfig:relative-error}, we show the relative estimation
error on full queries for each of the benchmarks and
methods. SafeBound has a similar range of errors as Postgres, but
guarantees that it never underestimates: its estimates in the figure
always lie above the center line. Traditional estimators frequently
underestimate by $10^3$ or more which is detrimental to query
optimization. Notably, additional optimisations such as
  Postgres2D and PostgresPK do not significantly alter the
  estimates. This implies that the errors primarily stem from the
  fundamental independence assumptions rather than a lack of
  statistics.  ML-based methods produce accurate estimates, but still
lack guarantees; NeuroCard is prone to significant
underestimates. The Simplicity system overestimates
  significantly due to its reliance on the max degree without
  conditioning on predicates.  Moreover, as discussed, its ``upper
  bound'' is not guaranteed: it returns a wrong upper bound on two of
  the queries of JOBLightRanges. Because it handles predicates by scanning the table at estimation time, PessEst has good estimates for queries with many predicates or with challenging predicates such as JOB-LightRanges and JOB-M. We provide a more detailed breakdown of estimation error by number of tables in the appendix.

\begin{figure*}[t]
    \centering
    \begin{subfigure}[b]{.4\textwidth}
    \includegraphics[width=.9\textwidth, height=110pt]{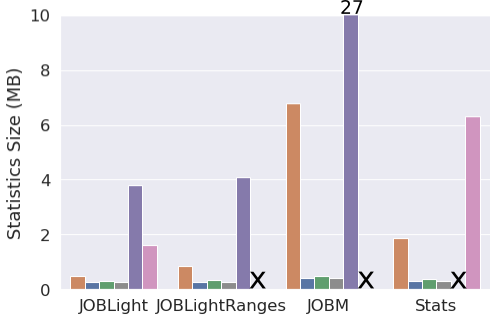}
    \caption{Across all benchmarks, SafeBound uses \textbf{3x-6.8x} less space than ML methods.}
    \label{fig:stat-size}
    \end{subfigure}
    \hfill
    \begin{subfigure}[b]{.51\textwidth}
    \includegraphics[width=.9\textwidth, height=110pt]{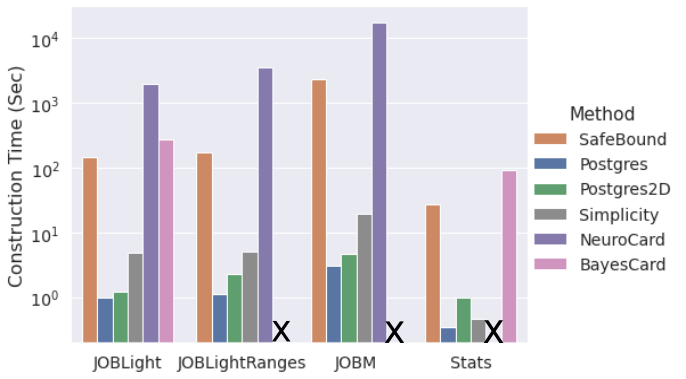}
    \caption{SafeBound achieves up to 17x lower offline statistics construction time than ML methods.}
    \label{fig:build-time}
    \end{subfigure}
    \caption{Statistic Size and Construction Time}
\end{figure*}

\begin{figure*}[t]
    \centering
    \begin{subfigure}[b]{.36\textwidth}
    \includegraphics[width=\textwidth, height=120pt]{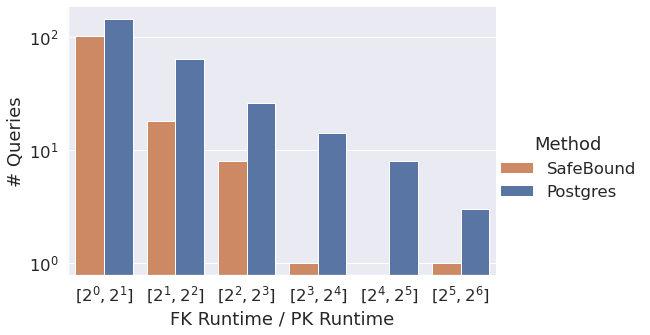}
    \caption{The frequency and magnitude of FK index performance regressions. SafeBounds estimates result in fewer and less severe regressions.}
    \label{fig:regressions}
    \end{subfigure}
    \hfill
    \begin{subfigure}[b]{.31\textwidth}
    \includegraphics[width=\textwidth, height=120pt]{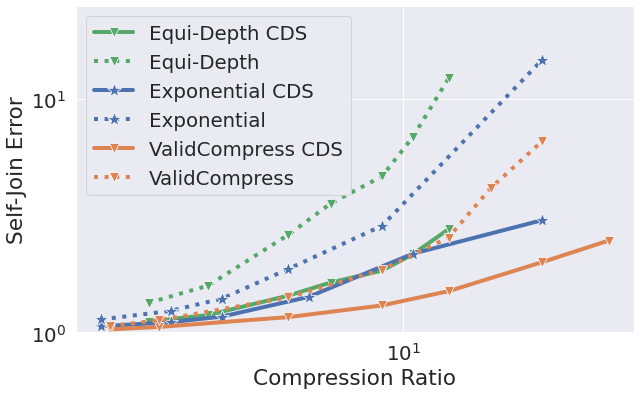}
    \caption{Modeling the CDS results in lower error than modeling the DS, and our two-pass algorithm outperforms heuristic methods.}
    \label{fig:segmentation-microbenchmark}
    \end{subfigure}
    \hfill
    \begin{subfigure}[b]{.31\textwidth}
    \includegraphics[width=\textwidth, height=120pt]{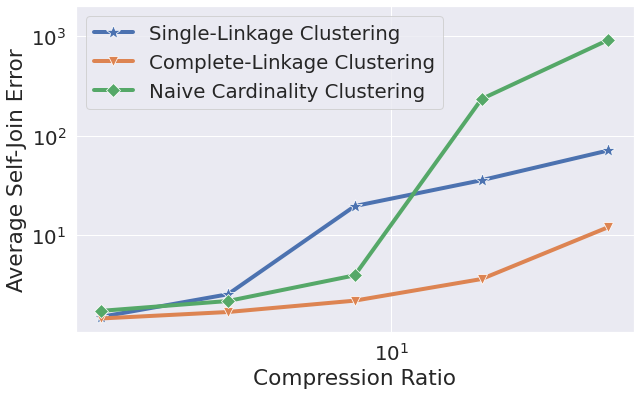}
    \caption{Complete-linkage clustering of CDS produces low error even at high compression rates.}
    \label{fig:clustering-microbenchmark}
    \end{subfigure}
    \caption{Micro-Benchmarks and Performance Regression Study}
\end{figure*}

\subsection{Planning Time, Memory, and Build Time}
\label{subsec:cost}

Figure~\ref{fig:planning-time} reports the planning times (Postgres'
optimization and the cost estimation time for all sub-queries) for all systems and
benchmarks. Postgres' efficient C implementation and use of dynamic programming in estimating sub-queries results in the fastest planning time. The Simplicity system achieves good planning times thanks to its straightforward bound calculation and reliance on Postgres's fast single-table estimates. PessEst requires scanning the base tables at
runtime when predicates are applied, resulting in 12x-420x slower inference. The ML methods, particularly Neurocard, perform inference on complex black-box models which leads to poor inference latency.
%
% The time required to perform training and inference on complex ML models remains another challenging obstacle to overcome. Large monolithic models like Neurocard particularly struggle to achieve the inference latency required by query optimizers in part because they take equally long to perform inference on a small query as they do on a large query. This is important because most sub-queries estimated during planning time are small, so methods need to scale down their computation to take advantage of simpler query structures.
%
SafeBound implements Algorithm~\ref{alg:FDSB:return}, which runs in
log-linear time in the size of the compressed degree sequences.
This results in much faster planning times than PessEst, BayesCard,
and NeuroCard across all benchmarks. 

Next, we turn to the memory footprint, which we report in Figure
\ref{fig:stat-size}.  Safebound's simple statistics and compression techniques allow 
it to achieve a compact memory footprint close to traditional methods like
Postgres. For instance, group compression results in 7.3-43x fewer degree sequences being stored across benchmarks. This results in statistics that are only 200KB larger than Postgres' for the JOB-Light benchmark and over 3x smaller for all benchmarks than
BayesCard and NeuroCard which rely on complex black box models. Simplicity relies on the statistics stored in Postgres and the max degree of each join column, resulting in a small memory footprint. PessEst does not operate on pre-computed statistics, so its space and
build time is not reported.

Finally, we considered the build time. SafeBound has 2-3.5x and 8-20x faster build
time than BayesCard and NeuroCard, respectively.  We note that SafeBound's build process
(Algorithm~\ref{alg:modeling}) is a series of aggregations over the
base data to create histograms and MCV lists, resulting in relatively fast
construction despite building from the full dataset rather than a
sample. Traditional estimators remain the fastest as they
perform similar aggregations to SafeBound but on a small sample of the data.
The string predicates of JOB-M result in longer build times for SafeBound and
NeuroCard because they require the calculation of tri-grams and
factorized columns, respectively. Postgres(2D), on the
other hand, does not keep statistics for LIKE predicates (instead
handling them via a magic constant), so the additional complexity does
not affect their time.

\subsection{Robustness Against  Regressions}
\label{subsec:regression}
When attempting to tune database instances, users often face
inexplicable performance regressions. Creating an index on a join
column is reasonably assumed by users to improve query performance,
but will frequently cause queries to run significantly slower. This
is primarily due to the query optimizer receiving cardinality
underestimates and optimistically making use of the index for a query
where a hash or merge join is faster. In figure
\ref{fig:regressions}, we show the frequency and severity of
regressions across all benchmarks when Postgres' internal estimates
are used vs SafeBounds cardinality bounds. While some regressions
still occur due to issues in cost modeling, SafeBound produces
half as many performance regressions across all benchmarks, $129$ to Postgres' $259$, and
they are half as severe on average, $1.7$x to Postgres' $3.3$x.

\subsection{Micro-Benchmarks} \label{subsec:micro:benchmarks}
In the following experiments, we evaluate how SafeBound's components perform individually as compared to alternatives. To do this, we calculate the error on self-join queries with equality predicates, specifically the \texttt{MovieCompanies} relation joining with itself on \texttt{MovieId} with and without an equality predicate on \texttt{ProductionYear}.

\noindent\textbf{Modeling the CDS rather than the DS reduces error by up to 20x.} By avoiding artificial inflation of the relations' cardinality caused by modeling the DS directly, SafeBound achieves significantly more accurate estimates. This can be seen on Figure \ref{fig:segmentation-microbenchmark} which shows the accuracy of various approximation methods for modeling the full CDS/DS of \texttt{MovieCompanies.MovieId} versus the compression ratio (\# distinct frequencies/ \# segments). The different colors correspond to different ways of choosing the segment boundaries for the approximation while the solid vs dashed lines correspond to whether the method models the CDS or the DS, respectively. As we would expect, every approximation method has lower error when applied to the CDS rather than the DS.

\noindent\textbf{The ValidApprox algorithm efficiently models the CDS.} Looking again at Figure \ref{fig:segmentation-microbenchmark}, we can compare the solid lines to get a sense for how different segmentation strategies affect the accuracy/compression trade-off. We compare against a couple of reasonable baselines, 1) an equi-depth strategy which segments the degree sequence into equal-cardinality segments 2) an exponential strategy which uses a geometric sequence for segment boundaries. Our two-pass algorithms outperform both baselines because it adjusts the size of buckets based on the skewness of the underlying degree sequence, taking into account both the importance of high frequency items and the long tail of the distribution.

\noindent\textbf{Complete-linkage clustering of CDSs provides low error at high compression ratios.} The experiment in Figure \ref{fig:clustering-microbenchmark} shows the effect of different clustering techniques. In this experiment, we joined the \texttt{MovieCompanies} relation with the \texttt{Title} relation according to their FK/PK relationship then calculate an MCV list for the \texttt{ProductionYear} attribute as described in Sec. \ref{subsec:choosing:statistics}. This results in $132$ CDS that we use to test various clustering methods, which cluster them into between $4$ and $64$ groups, and we represent each cluster with the maximum as in Sec. \ref{subsec:grouping:CDS:sets}. The error metric is then the average relative self-join error over all the original CDS when they are approximated using their cluster's maximum. In this case, compression ratio is defined as the number of original groups, $132$, divided by the number of clusters after compression.

We compare SafeBound's method, \textit{complete-linkage clustering}, with \textit{single-linkage clustering} and the naive method of grouping the functions into equal sized clusters  by cardinality. Across these methods, complete-linkage clustering results in lower error for all compression ratios. This is because naive clustering doesn't take into account the shape of the CDS and doesn't adaptively choose cluster sizes, and single-linkage clustering often produces long chaining clusters where one CDS dominates the maximum.

\begin{figure}[t]
    \centering
    \includegraphics[width=.42\textwidth]{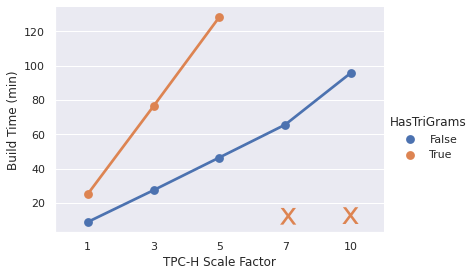}
    \caption{The pre-processing overhead increases linearly with the size of the dataset. [Python Prototype]}
    \label{fig:scalability}
\end{figure}

\subsection{Scalability}
\label{subsec:scalability}
To test how SafeBound scales to both more complex schemas and larger datasets, we experiment with the TPC-H benchmark at various scaling factors in Fig \ref{fig:scalability}. \footnote{We do not include this benchmark in our other experiments as its lack of skew, independence between columns, and independence across tables poorly represents real world distributions as discussed in \cite{DBLP:journals/vldb/LeisRGMBKN18}.} This benchmark has 14 join columns, 46 filter columns, and 9 PK-FK relationships over 8 tables, and we vary the scale factor from 1GB to 10GB.  We further build two versions of SafeBound: one where it constructs tri-gram statistics for LIKE predicates and one where it does not. This experiment shows that the construction time increases linearly in the size of the data. However, it also points to the inefficiency of the current python implementation in two ways. It runs out of memory when computing tri-grams for higher scale factors (as denoted by the X marks), and the build process is slower than expected given the simple underlying operations. This is in line with recent research which has shown an average of 29x worse performance when using CPython rather than C++ for a variety of applications due to dynamic type checking, interpreter overhead, and the global interpreter lock \cite{lion2022investigating}.

%% file: 6-conclusions.tex
\section{Limitations and Future Work}
\label{sec:limitations}
This work represents SafeBound, a first practical system for computing
{\em guaranteed} cardinality upper bounds: previous upper bound
systems either require significant query time computation (PessEst),
or do not provide {\em guaranteed} bounds (Simplicity).  SafeBound
achieves up to 80\% lower end-to-end runtimes than PostgreSQL across
workloads, and is on par or better than state of the art ML-based
estimators and pessimistic cardinality estimators.  Yet, the current
SafeBound prototype is likely not yet ready for a production system,
due to a few limitations, which we discuss here.

\noindent\textbf{Build Time:} Any upper bound system must, at some point, read all rows in the database. This places a hard lower bound on the build time for cardinality bounding systems. Therefore, future work to reduce this expense will need to exploit parallelism or take advantage of times when the system is already scanning the data (e.g. during bulk loading or query execution).

\noindent\textbf{Handling Updates:} A challenge that we leave open is
handling updates without recomputing the degree sequences.  We note
here that a degree sequences is essentially a
\texttt{group-by/count/order-by} query and could benefit from IVM techniques~\cite{DBLP:journals/vldb/KochAKNNLS14}.

%Degree sequences are challenging to update. When values are inserted into the column, there is no means of determining what segment of the degree sequence should be incremented. To get around this, it may be necessary to annotate the degree sequences with additional information about what join values are present in each segment. Alternatively, there may be a place for these techniques in versioned databases whose underlying data is semi-immutable.

\noindent\textbf{Multiple Predicates Per Table:} SafeBound handles the existence of multiple predicates on a single relation by taking the minimum of their induced CDS'. However, this results in an estimated cardinality equal to the selectivity of the most selective predicate. This could be inaccurate when faced with multiple moderately selective predicates which are jointly highly selective. Future work should consider ways of keeping lightweight statistics on combinations of filter columns in order to tackle this problem.

\textbf{Correctness and Accuracy:} SafeBound does not rely on {\em any} statistical assumptions about the data for correctness: it always provides a correct upper bound, regardless of the underlying data distribution, even in the presence of predicates. However, Sec.~\ref{sec:safebound} describes heuristics which affect the space/accuracy tradeoff, and queries which do not conform to the assumptions may see poor performance (e.g. joins on tables with very different skews or highly selective predicates). Further, like nearly all non-sampling methods, SafeBound does not provide a precise characterization of its error. Improving these heurstics or formally characterizing the tightness of SafeBound is an exciting avenue for future work.

%%% \section{Conclusions}
%%% \label{sec:conclusions}
%%% We have presented SafeBound, the first practical system for computing
%%% a bound on the cardinality of a query's output.  Unlike a cardinality
%%% estimator, SafeBound computes a {\em guaranteed} upper bound, and this
%%% can be beneficial for optimizers, especially when optimizing expensive
%%% queries, where cardinality underestimates can lead to costly mistakes.
%%% SafeBound is based on a recent theoretical
%%% result~\cite{DBLP:journals/corr/abs-2201-04166} that computes the
%%% cardinality bound from the degree sequences of all join attributes.
%%% Based on that theoretical framework it introduces degree sequences
%%% conditioned on predicates, a better compression method, and a
%%% practical compression algorithm.  SafeBound achieves up to 80\% lower
%%% end-to-end runtimes than PostgreSQL across workloads, and is on par or
%%% better than state of the art ML-based estimators and pessimistic
%%% cardinality estimators.
\section*{Acknowledgements}

This work was partially supported by NSF IIS 1907997, NSF-BSF 2109922, and a gift from Amazon through the UW Amazon Science Hub.